\documentclass[sigconf,10pt]{acmart}

\makeatletter
\def\@ACM@checkaffil{
    \if@ACM@instpresent\else
    \ClassWarningNoLine{\@classname}{No institution present for an affiliation}%
    \fi
    \if@ACM@citypresent\else
    \ClassWarningNoLine{\@classname}{No city present for an affiliation}%
    \fi
    \if@ACM@countrypresent\else
        \ClassWarningNoLine{\@classname}{No country present for an affiliation}%
    \fi
}
\makeatother

\microtypecontext{spacing=nonfrench}

\usepackage[english]{babel}
\usepackage{blindtext}

\renewcommand\footnotetextcopyrightpermission[1]{}
\acmConference[PACMNET]{}{June 2023}{To appear}%

\usepackage{times}  
\usepackage{hyperref}
\hypersetup{pdfstartview=FitH,pdfpagelayout=SinglePage}

\usepackage{enumitem}
\usepackage[aboveskip=0pt]{subcaption}
\usepackage{algorithm}
\usepackage[noend]{algpseudocode}
\usepackage{tabularx}
\usepackage{graphicx} 
 \usepackage{color}
 \usepackage{balance}
\usepackage{xspace}
\usepackage{url}
\usepackage{enumitem}
\usepackage{pbox}

\usepackage[font=footnotesize,labelfont=md,textfont=md]{caption}
\makeatletter
\algnewcommand{\LineComment}[1]{\Statex \hskip\ALG@thistlm \hskip\ALG@thistlm \(\triangleright\) #1}
\makeatother

\newif\ifpacor

\newif\ifshowpassive

\newif\ifEurosys
\newif\ifShortenForEurosys

\newcommand{\parheading}[1]{\vspace{0.05 cm}\noindent \textbf{#1}}
\newcommand{\SystemName}{Flock\xspace}
\newcommand{\Sys}{Flock\xspace}
\newcommand{\compactcaption}[1]{\vspace{-1.1em}\caption{#1}\vspace{-1.3em}}
\newcommand{\scompactcaption}[1]{\vspace{-0.2em}\caption{#1}\vspace{-0.9em}}

\usepackage{mathtools}
\usepackage{makecell}
\usepackage{boldline}

\def \thickhline {\Xhline{3\arrayrulewidth}}

\usepackage{amsmath,amsthm}
\usepackage{multirow,rotating} 
\usepackage{booktabs}

\newtheorem{theorem}{Theorem}
\newtheorem{lemma}{Lemma}
\newtheorem{definition}{Definition}

\setlist[itemize]{noitemsep,topsep=0pt}

\usepackage{titlesec}

\titlespacing{\section}{1pt}{0.6ex}{0ex}
\titlespacing{\subsection}{1pt}{0.6ex}{0ex}

\newif\ifshowpassive
\newcolumntype{P}[1]{>{\centering\arraybackslash}p{#1}}
\newcolumntype{M}[1]{>{\centering\arraybackslash}m{#1}}

\newif\ifGeneralJLE

\DeclareMathOperator*{\argmax}{arg\,max}

\begin{document}

\setlength{\skip\footins}{1pt plus 2pt minus 0pt}
\setlength{\topsep}{0.2em}
\setlength{\floatsep}{1.0pt plus 1.0pt minus 0.0pt}
\setlength{\intextsep}{1.0pt plus 1.0pt minus 0.0pt}
\setlength{\textfloatsep}{1.0pt plus 1.0pt minus 0.0pt}
\setlength{\dblfloatsep}{1.0pt plus 1.0pt minus 0.0pt}
\setlength{\dbltextfloatsep}{1.0pt plus 1.0pt minus 0.0pt}
\setlength{\abovecaptionskip}{0.1pt}
\setlength{\belowcaptionskip}{0.1pt}
\setlength{\abovedisplayskip}{0.2pt}
\setlength{\belowdisplayskip}{0.2pt}
\setlength{\abovedisplayshortskip}{0.2pt}
\setlength{\belowdisplayshortskip}{0.2pt}

\title{Flock: Accurate network fault localization at scale}


\author{Vipul Harsh, Tong Meng, Kapil Agrawal, P. Brighten Godfrey 
}
\affiliation{
\institution{University of Illinois at Urbana-Champaign}
}


\begin{abstract}
Inferring the root cause of failures among thousands of components in a data center network is challenging, especially for ``gray'' failures that are not reported directly by switches. Faults can be localized through end-to-end measurements, but past localization schemes are either too slow for large-scale networks or sacrifice accuracy.
We describe Flock, a network fault localization algorithm and system that achieves both high accuracy and speed at datacenter scale.  Flock uses a probabilistic graphical model (PGM) to achieve high accuracy, coupled with new techniques to dramatically accelerate inference in discrete-valued Bayesian PGMs.
Large-scale simulations and experiments in a hardware testbed show Flock speeds up inference by >$10^4$x compared to past PGM methods, and improves accuracy over the best 
previous datacenter fault localization approaches, reducing inference error
by $1.19-11\times$ on the same input telemetry, and by $1.2-55\times$ after incorporating passive telemetry. 
We also prove \Sys's inference is optimal in restricted settings.

\end{abstract}

\maketitle



\section{Introduction}


Datacenters often comprise of tens of thousands of network components. Failures in such large networks are common, arising due to software bugs, misconfiguration, and faulty hardware, among other reasons~\cite{achillesHeel}.  In many cases, a device will directly report a failure, e.g., a switch may report that one of its line cards is non-responsive or that an interface has a certain packet loss rate. Network operators utilize monitoring software to collect these metrics and raise alerts. 
However, datacenters also experience significant network downtime and SLO violations from \emph{gray failures} whose root cause is obscure~\cite{achillesHeel,everflow}.  For example, the reason for poor performance of a distributed service could be a link silently dropping a small fraction of packets without updating switch counters~\cite{netbouncer}, or a driver bug in a virtualized firewall. Diagnosing such performance anomalies is very hard for network operators~\cite{007}. With programmable switches, more advanced monitoring is possible~\cite{netseer,omnimon}, but these methods either do not eliminate gray failures~\cite{INT2.1} or come at a high cost in switch resources~\cite{netseer}, and require deployment of programmable switches which is not generally available.





An alternate approach, which is the domain of this paper, is to infer the root cause via end-to-end measurements, which we refer to as \emph{fault localization}~\cite{007,netbouncer,pingmesh,netnorad,KDD-Fault-Localization,simon}.
Fault localization has been deployed by large cloud providers~\cite{sherlock,pingmesh}. At the heart of fault localization is an \emph{inference model and algorithm} that uses end-to-end observations (e.g., packet loss rate or latency of TCP connections) to infer a set of faulty components (links or switches).  The key challenge is to do this both accurately and quickly.

The most powerful class of inference techniques builds a probabilistic graphical model (PGM) and performs a form of maximum likelihood estimation (MLE): finding the set of components that, if they failed, maximizes the probability of having produced the given end-to-end observations. 
An early such system was Sherlock~\cite{sherlock}, whose primary context was inferring faults among dependent services. However, deriving the MLE for a PGM can be computationally intensive. With $\leq k$ failures among $n$ components, the solution space is exponentially large in $k$ ($O(n^k)$).
In a datacenter with millions of TCP flows, switches, links, etc. and multiple simultaneous failures, Sherlock's MLE can require several hours.


Therefore, fault localization schemes intended for datacenter networks move away from PGMs to other techniques -- using scores to rank links~\cite{007,KDD-Fault-Localization} or solving for drop rates via a system of equations~\cite{netbouncer}. As we will show, these compromise accuracy and flexibility of PGMs in favor of performance. 

In this paper, we present \Sys, a fault localization system for datacenter networks that seeks to maximize accuracy, with sufficiently high speed (i.e., seconds).  \Sys's core innovation is a novel MLE inference algorithm for PGMs that offers substantial acceleration for the kind of models encountered in fault localization, leveraging two key acceleration approaches: (i) a technique we call \emph{joint likelihood exploration} maintains an array of hypotheses that it can update en masse to find the likelihood of a set of \emph{new but similar} hypotheses, more quickly than computing their likelihoods individually from scratch and (ii) we use a greedy algorithm which builds its solution link by link; this part of the algorithm is simple, of course, but importantly, we prove a sufficient condition for optimality and verify through experiments that it does find the MLE in practice. These two optimizations each individually provide asymptotic speedups, and together allow \Sys to use a PGM at scale.  The result is that \Sys is several orders of magnitude faster than past PGM-based fault localization~\cite{sherlock}, and is substantially more accurate than past non-PGM-based fault localization~\cite{007,netbouncer}, on the same input data.


Moreover, the PGM-based approach allows \Sys to use different types of input telemetry. Recent datacenter fault localization schemes~\cite{netbouncer,007} use observations of \emph{active probes} of the network (which can be constructed to have known paths and uniform distribution) but do not incorporate \emph{passive flow monitoring}, i.e., observations of all ongoing traffic, obtained via NetFlow, IPFIX, or INT.
\footnote{\cite{007} incorporates only a limited amount of passive monitoring; \S \ref{sec:existing-approaches}.} 
The large volume of passive data makes it potentially informative.  But including passive monitoring would be hard for non-PGM methods because it requires more discerning modeling and inference to handle skew in traffic patterns and path uncertainty.\footnote{Most data centers use non-deterministic ECMP multipath routing, so that only a \emph{set of possible paths} is known when flows are monitored with NetFlow/IPFIX.  Paths can be known with INT-based monitoring, but INT is not generally deployed.
} Although PGMs are generally more flexible in incorporating data of different types, it would be hard for past PGM approaches because of computational difficulty, due to the immense number of flows and because a flow with $10$ possible paths is roughly $10\times$ costlier to model than a flow with a known path.  \Sys's combination of flexible PGM-based modeling and speed enables it to utilize  passive information.

In summary, this paper's key contributions are as follows:

\begin{itemize}[leftmargin=7.0pt]
    
    \item \textbf{Inference algorithm.} We develop \Sys, a new fast PGM MLE inference algorithm for the type of PGM necessary for fault localization, namely discrete-valued Bayesian PGMs (\S~\ref{sec:inf-alg}). We analyze a sufficient condition for this algorithm's accuracy on Flock's model, providing intuition for why Flock works well in practice (\S~\ref{greedyJustification}).

    \item \textbf{System implementation.} We implement the \Sys system (\S~\ref{flock design}), a new end-to-end fault localization system.  The \Sys algorithm forms the heart of the \Sys system, allowing it to employ a PGM and naturally incorporate various kinds of dependence and uncertainty such as unknown paths.
    
    \item \textbf{Evaluation suite.} We create an open
    evaluation suite~\cite{githubRepoFlock} for fault localization, which includes (a) implementations of  algorithms from NetBouncer~\cite{netbouncer}, 007~\cite{007}, Sherlock~\cite{sherlock}, and \Sys, (b) an implementation of end-host telemetry agents and a collector, (c) telemetry data for six different fault scenarios from a simulated data center and a hardware testbed, and (d) scaling tests. We believe this suite 
    is of independent interest, as it 
    is the first such open data set and expands on the fault scenarios evaluated by past work. 
    


    \item \textbf{Performance evaluation.} 
    For a Clos network with 88K links and 9.5M flows, 
    \Sys is empirically $>10^4\times$ faster than Sherlock's PGM-based method~\cite{sherlock}, scanning \textasciitilde3.5M hypotheses in 17 sec, while achieving the same or better inference results (\S~\ref{runningTimeSection}). In fact, \Sys is $\approx 4.5\times$ faster than the non-PGM approach of NetBouncer~\cite{netbouncer} on the same input. 007~\cite{007} is the fastest of the lot, but its time savings (<1 sec) is not a good tradeoff with accuracy.
    \item \textbf{Accuracy evaluation.} With the same (active probe) input measurements as past work, \Sys reduced inference error by $1.8-8\times$ compared to 007 and by $1.19-11\times$ compared to NetBouncer. 
    \item \textbf{The value of passive information/INT.} 
    Incorporating passive monitoring reduced error even further, by up to $5.3\times$ compared to \Sys with only active monitoring.
    We also evaluate the value of INT telemetry input. 

\end{itemize}

\section{Background and Motivation} \label{motivation}

\subsection{When fault localization is useful}
\label{sec:fl-uses}
Ideally, network faults are directly reported by the faulty component; e.g., switches typically track interface utilization, packet drops, up/down status, queue length, etc. These metrics are commonly collected from network switches via SNMP~\cite{rfc1157}, polling~\cite{CiscoSNMPInterval,SolarwindsSNMPInterval}, or streaming telemetry~\cite{CiscoStreamingTelemetry,AristaNetworkTelemetry}.

Fault localization becomes useful when such direct monitoring is not enough.  Problems that are not directly reported by the faulty components are known as \emph{gray failures}~\cite{achillesHeel}.
For example, silent inter-switch or inter-card drops~\cite{netseer,netbouncer,007,FbLocalization,everflow} are extremely challenging to detect, constituting 50\% of faults that took  >3 hours to diagnose in \cite{netseer}. Other examples of gray failures include corruption in TCAM-based forwarding tables causing black holes~\cite{maxCoverage,everflow} or loss~\cite{CiscoBugCSCvn56156}, and a misconfigured switch causing high latency~\cite{everflow} (see~\cite{netseer} for more cases).  gray failures also occur in the numerous software packet processing components present in modern data centers. Software bugs can silently drop or corrupt packets in host virtualization~\cite{VMwareCorruption,VMwareLoss}, server software~\cite{everflow}, and virtualized network functions like software firewalls~\cite{paloaltobugs}.  All the above faults can be ``silent'' (the device does not 
realize the error occurred). Further, the symptoms could manifest at a different location than the problem, e.g. packet data could be corrupted by an intermediate switch but the corruption is discovered only at the receiver. End-to-end observations are well suited to detect such problems (\S~\ref{failureScenarios}). 



Another alternative is to use programmable switches to obtain more information, as in Omnimon~\cite{omnimon}, FANcY~\cite{sigcomm22fancy} (for ISPs), or NetSeer~\cite{netseer} which runs a packet sequencing protocol across neighboring hops to find silent drops. This can be quite accurate, though it comes at the cost of significant switch resources (\textasciitilde100\% overall PHV usage and 40\% ALU usage~\cite{netseer}).  But more importantly, although use of programmable switches is growing, they still have very limited deployment (13\% estimated market share for 2023~\cite{progSwitchMarket}).  Both programmable and traditional switching environments are valuable use cases, but \emph{the latter is the target of this paper;} schemes like~\cite{omnimon,netseer,sigcomm22fancy} are out of the scope of our work.  As an exception, we will consider the use of data collected with In-band Network Telemetry (INT)~\cite{INT2.1,ben2020pint}, which does not directly report gray failures, but does record packets' paths.  INT can be implemented with programmable switches, but similar path data can be obtained other ways; see \S~\ref{information kinds}.

Thus, it is very hard to guarantee that every packet processing element detects all faults locally (indeed, the end-to-end principle~\cite{saltzer1984end} applies here). Also, even if a fault is reported, the operator may want a way to cross-check. For these reasons, we see fault localization based on end-to-end observations as an important tool in infrastructure engineers' toolbox for the foreseeable future. 

\subsection{Problem setup and goals}
\label{sec:goals}

\noindent The input to a network fault localization algorithm is the network topology, and input telemetry (flow measurements). Each flow measurement includes one or more metrics (TCP loss rate, mean latency, throughput, etc.) and a set of paths through the topology that the flow may have traversed. Depending on the monitoring method, this set may have size one (the exact path is known) or greater than one (typically, a set of possible ECMP paths is known).  Given this input, the fault localization algorithm should output a set of links or devices it believes to be faulty, while meeting two goals.

\parheading{Performance}: 
Network operators often have to resolve a reported problem quickly.
For example, a managed service from BT has a 15 minute response time in its SLA~\cite{BTManagedWANContract} and 
Gartner chose a threshold of 3 minutes in defining ``real time'' network data analysis~\cite{GartnerRealTime}.  
Thus, fault localization within minutes is critical and within a few seconds is ideal.  

\parheading{Accuracy:} 
False positives can bury true problems among several alerts~\cite{cloudDatacenterSdnMonitoring}. 
False negatives could send engineers down the wrong track of investigation. There may be a tradeoff between accuracy and performance.  
As long as results are available within a few minutes, accuracy (minimizing false positives and false negatives) is of primary importance.  

\subsection{Existing fault localization approaches}
\label{sec:existing-approaches}

\noindent Several past approaches address above goals, with different types of \emph{input telemetry} and different \emph{inference algorithms}.

\parheading{Input telemetry:} Administrators commonly~\cite{GartnerRealTime} use passive monitoring of flows via NetFlow~\cite{claise2004cisco} or IPFIX~\cite{claise2013specification} to understand overall network health. 
However, this has not been relied on for automated fault localization because vendor-specific ECMP hashing obscures flows' exact paths. 

Recent approaches use active probes with known paths.  NetBouncer~\cite{netbouncer} sends probes uniformly from hosts to core switches in a Clos topology, via special switch support.  007~\cite{007} uses active probes with assistance from passive monitoring: end-host agents monitor production traffic, and when they detect flows with anomalous performance, the agents traceroute the flagged flows' paths and report the flagged flows' metrics for analysis.  007 does not incorporate passive monitoring of non-flagged flows. 
In both systems, the volume of active probes is limited (which assists inference performance, and minimizes 
 host/network overhead,
 and the exact path of monitored flows is known (which assists accuracy). But active probes do not eliminate all uncertainty: even if the path is known, the culprit link(s) are not; and flows can experience packet loss or latency on non-faulty links (e.g., due to congestion). Hence, good inference is still needed. Further, active probes often take a different data path than regular traffic and may fail to reproduce problems faced by production flows~\cite{007}.

INT has, to our knowledge, not been specifically used by past work 
for end-to-end fault localization.  INT is similar to passive NetFlow telemetry in that it can observe a large volume of actual application traffic, if deployed for all flows.  However, like active probes, it can trace the traversed path. It also does not eliminate all uncertainty (e.g., a silent packet corruption will not trigger an INT action, and even if the path is known, the faulty link is not).

We observe that \emph{different deployments are likely to have different available information.}  007 requires host support and NetBouncer and INT require switch support. INT is now available in some switches, but is not deployed in most networks, and might be deployed selectively.
Furthermore, this technology landscape may evolve. A scheme that is flexible in its input telemetry is thus preferred.


\parheading{Inference algorithms:}
Sherlock~\cite{sherlock}, NetSonar~\cite{netsonar} and Shrink~\cite{shrink} use a PGM with a form of maximum likelihood estimation (MLE), but are far too slow. 
Sherlock targeted a small use case (358 components), and NetSonar, which uses Sherlock's inference, targeted smaller inter-DC networks. In a data center network with tens of thousands of components, they require several hours (\S~\ref{runningTimeSection}), even with $K\leq2$ concurrent failures (they scan $O(n^K)$ possibilities, where $n=$ number of components).
Furthermore, inevitable concurrent failures~\cite{netbouncer,007,FbLocalization} may make $K>3$ essential.

Hence, state-of-the-art data center fault localization schemes move away from PGM-based MLE. 007~\cite{007} uses a scoring function to rank links while NetBouncer~\cite{netbouncer} optimizes an objective function to solve for drop rates. These are reasonably fast, but as we will see, they fall short on accuracy.



\parheading{Summary:} Our goal is to localize ``gray'' failures in datacenter networks, using end-to-end information, achieving as high accuracy as possible within roughly a few seconds (including monitoring and inference), making best use of available information, including active probes or INT (where paths are known) and passive monitoring (exact paths are unknown). 

\section{Flock Design} \label{flock design}

We present \Sys in three components.  First (\S~\ref{sec:data-collection}), \Sys monitors end-to-end flows at the endpoints (e.g. hosts) of the network. The monitored flow observations, comprised of metrics from active probes and (when available) passive flow monitoring and/or INT, are then sent to a central collector. Next, the collector periodically constructs a probabilistic graphical model (PGM) from the flow observations (\S~\ref{inference graph model}). This PGM captures uncertainty in how the observations may have resulted from underlying network components.  Finally, Flock periodically performs MLE inference (\S~\ref{sec:inf-alg}) on the PGM model, searching through the exponentially-large hypothesis space (i.e., sets of possible faulty components) to find the hypothesis that maximizes the probability of the observed flow metrics in the model.

The first two components (monitoring and the PGM model) are not the significant contributions of this paper.  Our PGM is similar to that of Sherlock and NetSonar~\cite{netsonar}, with useful adaptations to fit our environment.  We describe those pieces for completeness. Our key contribution lies in the MLE inference algorithm, comprised of multiple optimizations to make it scale. 

\subsection{Flow monitoring}
\label{sec:data-collection}

Monitoring of end-to-end flows may occur in an agent process running on hosts, in the hypervisor of a virtualized data center, or potentially at edge/top-of-rack switches (e.g., with INT).  Our inference algorithms are agnostic to where exactly these measurements come from. 
Still, for concreteness, we describe the operation of an endpoint monitoring agent (\S\ref{implementation}).  

The agent periodically actively probes the network  
and may optionally passively observe performance of ongoing flows. Note that 007 also observes ongoing traffic, to decide what active probes to launch. Metrics from both active and passive monitoring are aggregated by flow, and optionally randomly sampled to reduce volume.  Periodically, the agent sends these reports to the collector.  For the rest of this section, we assume flow reports include RTT, source, destination, retransmissions, packets sent and the path if known via active probing/INT, else the set of possible paths.

\subsection{Inference graph model} \label{inference graph model}

After telemetry is collected, \Sys builds a probabilistic model of the network using two inputs: (1) telemetry reports as described above, potentially including active probes, INT, or passive telemetry, and (2) the network topology and routes.  
The latter could be provided by an SDN controller or a topology discovery tool and is used to determine a \emph{set} of paths (typically, ECMP paths) that each flow may have traversed in the specific case of passive data with unknown paths.



Flock employs a probability model based on a Bayesian network that utilizes flow level metrics.  A Bayesian network is a probabilistic graphical model that defines the probability distribution of a set of observed variables in terms of a set of unknown (or \emph{hidden}) variables. It consists of a labeled directed acyclic graph (DAG) where each node represents a random variable, unknown or observed, whose distribution can be specified completely as a function of its immediate ancestors in the DAG. The goal of the inference is to estimate the unknown variables (failed/not failed status of each component) given the observed variables (retransmissions, RTT, packets sent etc. associated with each flow).

\begin{figure}
  \begin{minipage}[b]{0.3\textwidth}
    \centering
    \hskip 0 cm
  \begin{subfigure}[b]{\linewidth}
    \includegraphics[width=0.98\textwidth]{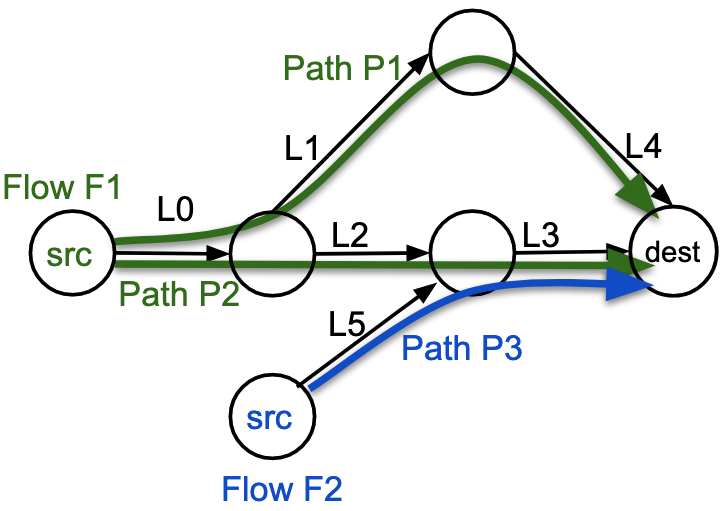}
    \caption{Network}
    \label{fig:flow_paths}
  \end{subfigure}
  \end{minipage}
   \hskip 0 cm
    \begin{minipage}[b]{0.3\textwidth}
    \centering
    \begin{subfigure}[b]{\linewidth}
    \includegraphics[width=1.1\textwidth]{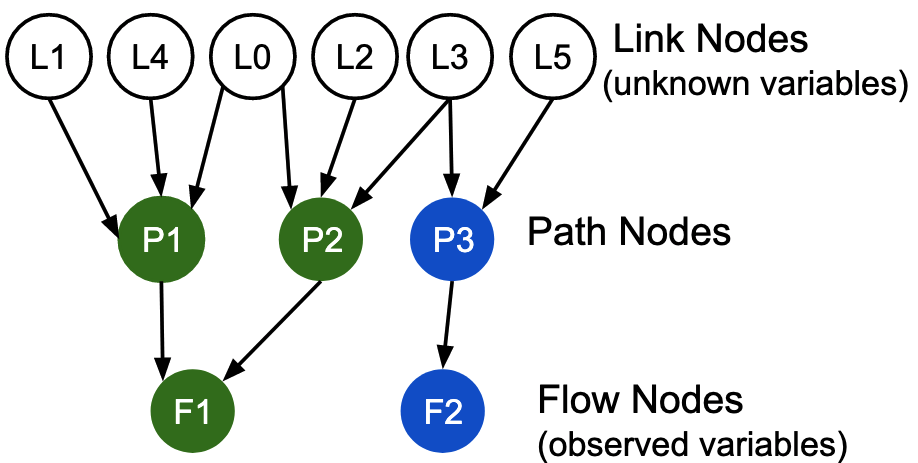}
    \caption{PGM graph}
    \label{fig:flow_bayesian_subgraph}
  \end{subfigure}
  \end{minipage}
  \caption{(a) A network with flows F1 and F2 traversing several possible paths; (b) the corresponding PGM model.}
    \label{fig:PGM_model}
\end{figure}


We use a 3-layer graph model. Fig.~\ref{fig:PGM_model} shows an example network with two flows and its corresponding model.  Formally, the conversion from datacenter network to model is as follows. The \textbf{top layer} nodes in the graphical model represent individual links in the datacenter, referred to as \textit{link-nodes}. Each link-node is a hidden binary variable which is 0 if the link has failed and 1 otherwise.  
The \textbf{intermediate layer} consists of nodes corresponding to paths in the datacenter, referred to as \textit{path-nodes}. For each link $\ell$ in path $p$, there is an edge from the link-node of $\ell$ to the path-node of $p$ in the Bayesian graph. A path-node represents an intermediate binary variable which is 0 if the path consists of a faulty link and 1 otherwise. Finally, the \textbf{bottom layer} consists of \textit{flow-nodes} -- one for every flow -- which are the observed variables. A flow-node has an edge from each path-node in the path-set of that flow.  We define the value of the flow-node variable as the number of \textit{bad} packets -- packets which experienced a problem.
We describe two ways of setting the bad packets variable:
\begin{itemize}[leftmargin=*]
    \item Per packet analysis: For packet loss and corruption, we set the number of \textit{bad} packets as the number of retransmissions which serves as a proxy for lost packets.
    \item Per flow analysis: To capture symptoms of high latency, we use a ``per-flow'' analysis which is in effect a special case of the per-packet model where the number of packets sent is 1, and the number of bad packets is set to 1 if the flow's RTT is above a threshold and 0 otherwise.
\end{itemize}


 
\noindent \Sys's model assumes a flow $F$ is routed via ECMP; $F$ takes one of $w$ paths chosen uniformly at random, and packets experience problems independently and uniformly at random.  Thus, the probability of observing $r$ bad packets out of the $t$ sent can be given as:
\begin{equation}\label{likelihood equation}
P[F = (r, t)]=\frac{1}{w} \sum_{i=1}^{w} (1 - \gamma_i) p_b^{r} (1 - p_b)^{t-r} + \gamma_i p_g^{r} (1 - p_g)^{t-r}
\end{equation}

\noindent where $\gamma_i$ is the value of the $i^{th}$ possible path of $F$ ($\gamma_i = 0$ if a failed link is on the $i^{th}$ possible path and 1 otherwise).  $p_g$ and $p_b$ are model hyperparameters: $p_b$ represents the probability of a packet experiencing problems when taking a bad path (i.e., with at least one faulty link), and $p_g$ represents the probability of a packet experiencing problems on a good path (no faulty links). Intuitively, a packet going through a failed link is much more likely to experience a problem, hence $p_b >> p_g$. 
\S~\ref{sec:paramCalibration} describes how to pick $p_g$ and $p_b$. Equation~\ref{likelihood equation} can also be adapted to include path weights, like in WCMP~\cite{WCMP}.

A \emph{hypothesis} is an assignment $H\in \{0,1\}^n$ for all link-nodes.  Equivalently, we can think of $H$ as a \emph{set of links} that are deemed to be failed, with all other link-nodes being not failed. 
The goal of the inference is to recover the hypothesis that consists of all truly failed links and only those links. 
Conditioned on a hypothesis $H$, the probability of the set of flow observations taking on the observed assignment of values (a certain number of bad packets $r_i$ out of $t_i$ packets sent, for each flow $i$) is simply the product of probabilities of all individual flow probabilities:

\begin{align*}
& P[F_1, F_2, \ldots, F_m|H]  = \prod_{F_i\in \text{flows}} P[F_i = (r_i, t_i)|H\big] = \prod_{F_i\in \text{flows}} P[F_i | H\big]
\end{align*}
where $F_i$ is shorthand for the event that flow $i$ takes on the observed metric values, i.e., $F_i = (r_i, t_i)$.




\parheading{Incorporating Priors.} We assign a prior belief about failures by assuming that, a priori, any link can fail with probability $\rho$. 
The priors reduce the false positive rate by effectively assigning a lower prior to hypotheses with more links, thus favouring hypotheses with fewer failed links. If hypothesis $H$ contains $|H|$ candidate failed links and there are $n$ total links, then the likelihood of $H$ after incorporating priors is given as:
\vspace{-0.07in}\[Prior * \prod_{F_i\in \text{flows}} P[F_i|H];\ \ Prior = \rho^{|H|}(1-\rho)^{n-|H|}\]




\parheading{Model extensions.} The top layer nodes can include other component types besides links; we add device nodes, treating a device as another component in a flow's path, exactly as links. 
We found that a device prior that is $5\times$ larger on log-scale, worked well in practice, as it forces \Sys to detect a device failure only when there is stronger evidence for it than a link failure.
Other components (line cards, racks, pods etc.) can be modeled in a similar way, but is beyond the scope of this paper.



\parheading{Model Intuition.} The model effectively incorporates several kinds of uncertainty. Given an observation of bad packets in a flow, we may not know what path is responsible (modeled via flow nodes having multiple path parents).  Even if we do know the path, as we do for active probes and INT, we don't know what link is responsible\footnote{While INT can capture per-hop metrics, it may not detect where a gray failure happens.  For example, a silent corruption of packet data likely would not be detected until the packet reaches the receiver's host stack.} (this is modeled via path nodes having multiple link parents).  Even if we know what link is responsible, an observed bad packet might or might not mean there is a faulty link (this is modeled via both good and bad links having non-zero probability of bad packets).

\parheading{Differences from Sherlock's PGM.} Sherlock was intended to model application-level failures, and thus includes elements that we don't need such as services and load balancers. Sherlock uses three node states -- working, failed, and partially working; we omit the last, as Flock models some packet loss even for working links. Starting from Sherlock's model and making these changes results in a PGM that is very close to \Sys's, except for the probability formulations.

\parheading{Model Assumptions.} Like any other model~\cite{sherlock,shrink,netsonar}, ours has assumptions: fixed packet fail probabilities (as in~\cite{shrink}), classifying paths as failed on just the absence/presence of at least one failed link (as in~\cite{netsonar}), and packets getting affected independently (as in~\cite{netsonar,shrink,sherlock}). We tried different assumptions -- one with variable fail rates obtaining MLE using Nesterov's Gradient Descent algorithm, but it was too slow; another variation that treats paths differently depending on their number of failed links, but its accuracy was worse.\footnote{Note that Flock still localizes multiple failures on a path since there are flows that transit one failed link, but not the other.}
We present the model that gave the best results in experiments based on real environments that don't adhere to any model assumptions.  We also give theoretical (\S~\ref{greedyJustification}) evidence supporting Flock's effectiveness with these assumptions. 
Finally, it's important to keep in mind that the model does not need to match reality perfectly, \emph{it only needs to be ``close enough'' that the most likely explanation in the model is the right one}. See Fig:~\ref{fig:example} in appendix for an illustration of how the PGM-inference can localize more accurately than past schemes.

\subsection{Inference algorithm} \label{sec:inf-alg}
We describe the inference algorithm next. For ease of exposition, this description only refers to links as components that can fail. Devices are treated exactly in the same way as links and our implementation handles both links and devices.

Recall that a hypothesis $H$ is a candidate set of  failed links.
From the model,  we can compute the probability of the flow variables taking their observed values (number of bad/sent packets for each flow) given $H$; this is the likelihood of $H$. We denote this likelihood as $L(H) = P[F_1|H]P[F_2|H]\ldots P[F_m|H]$. The maximum likelihood estimator  $\mathcal{H}$ is the hypothesis that maximizes $L(H)$, or equivalently log likelihood $LL(H) = \log L(H)$, that is  $\mathcal{H} =  ~ \underset{H\subseteq  \text{links}}{\text{arg max\ \ }} LL(H)$.

We normalize all likelihoods by the likelihood of the no-failure hypothesis (i.e., $H_{0} = \{\}$) to cancel out any flow whose path set does not include any failed links in a hypothesis $H$.  

The goal of MLE inference is to compute $\mathcal{H}$. Simply computing the likelihood of each possible hypothesis would be impractically slow because there are $2^n$ hypotheses, where $n$ is the number of links. Sherlock~\cite{sherlock} and NetSonar~\cite{netsonar} limit max concurrent failures to $k$, reducing the search space to $O(n^k)$ hypotheses. However, in our setting, this is still far too slow, even for $k\!=2$ (\S~\ref{runningTimeSection}). Further, a datacenter can have many concurrent failures making $k\!>3$ important.


We introduce two algorithmic techniques to accelerate MLE inference in PGMs. \emph{Greedy search} reduces the number of hypotheses examined. While Greedy MLE is a simple idea, our main contribution is showing that it finds good solutions even after narrowing the hypothesis space: we provide theory (\S~\ref{greedyJustification}) and experiments (\S~\ref{experimentalSetup}). \emph{Joint likelihood exploration} (JLE) is a new algorithmic technique, that reduces the time required per examined hypothesis.  
Both techniques provide significant speedups individually and 
together speedup the inference by several orders of magnitude (\S~\ref{runningTimeSection}).


\parheading{Greedy Search}:
We start from the no-failure hypothesis and extend it one link at a time.  Specifically, we maintain a current hypothesis $H$.  Initially, $H = \{\}$. In each iterative step, we scan over each link $l \notin H$ and calculate $LL(H \cup \{l\})$.  If one of these log likelihoods improves over $LL(H)$, we set $H := H \cup \{l^*\}$ where $l^*$ is the link offering the biggest improvement, and continue iterating.  When no added link failure improves the log likelihood of the current hypothesis $H$, the search terminates and returns $H_{greedy} = H$.


There is still a performance challenge with Greedy MLE. Each iterative step requires evaluating close to $n$ hypotheses (specifically $n - |H|$) to find $l^*$. Even with 40 cores, greedy search took over 3 hours for a medium-sized datacenter (\S~\ref{runningTimeSection}). This motivates our second key optimization.

\vskip 0.2 cm

\noindent \textbf{Joint likelihood exploration (JLE)}: We devise an additional acceleration technique for inference algorithms for PGMs with discrete variables. We use it to speed up each iteration of the greedy algorithm by a $O(n)$ factor, where $n$ is the number of components (links and switches). Note that greedy+JLE produces the exact same solutions as greedy.


Suppose we are given the current best hypothesis $H$ which has the maximum likelihood among hypotheses searched till now. Joint likelihood exploration is a technique to quickly explore all ``neighbors'' of $H$ -- all assignments that are different from $H$ in the inclusion or exclusion of exactly one link. Note that there are $n$ such neighbor hypotheses of $H$.

\vskip -0.1in
\begin{definition}
Let $H\oplus l$ denote the hypothesis obtained by flipping the status of link $l$ in $H$, i.e., if $l\in H$ then $H\oplus l = H \setminus \{l\}$ and otherwise $H\oplus l = H \cup \{l\}$.
Let $\Delta_H(l)\! = LL(H\oplus l)- LL(H)$ represent the difference in log likelihoods of hypotheses $(H\oplus l)$ and $H$.
\end{definition}
\vskip -0.1in
\parheading{JLE Intuition}: We explain JLE by showing how it accelerates the Greedy algorithm (although it can also accelerate exhaustive search).  In each iteration, Greedy computes each $LL(H \cup \{l\})$ for each $l$, to find the link $l^*$ offering the most improvement.  Note that maximizing $LL(H \cup \{l\})$ is equivalent to maximizing the difference in log likelihoods, $\Delta_H(l)$.  So, in each iteration of Greedy, we could compute an $n$-element array $\Delta_H$ whose elements are the values $\Delta_H(l)$ for each link $l$, and then scan the array to find the largest value. This will find the same $l^*$ as in the original greedy algorithm.


But how do we compute the array $\Delta_H$? If we do it in the obvious way, by iterating over each $l$ and directly computing $LL(H\oplus l)- LL(H)$, then this is nearly identical to the original greedy algorithm which computed $LL(H \cup l)$, with no appreciable change in runtime. What JLE offers is \emph{a faster way to compute $\Delta_H$}, with careful algorithm engineering. This involves two somewhat different types of computation: quickly preparing the array $\Delta_{H_0}$ for the first iteration; and quickly iteratively \emph{updating} the array for each subsequent iteration.

To initially create $\Delta_{H_0}$, note each entry $\Delta_{H_0}(l)$ represents the difference in log likelihood due to failing link $l$, compared to no failures.  To compute these differences, we only have to look at the effect on flows whose paths intersect $l$.  Furthermore, there is an important opportunity for memoization in computing the log likelihood-difference formula across different links $l$: the effect on a flow's likelihood depends only on the number of failed paths, not the specific failed links. 

Now suppose we have an existing $\Delta_H$ and the search algorithm is about to move from hypothesis $H$ to hypothesis $H'$ in its next iteration.  We need to compute the array $\Delta_{H'}$. To do this, we track the \emph{difference in the difference arrays} ($\Delta_{H}$ vs. $\Delta_{H'}$) rather than directly computing the difference in likelihoods of $H$ and $H\oplus l$ for every $l$. The key insight here is that each entry of the difference array $\Delta_H$ can be written as a sum of contributions for all flows and only some of the terms need to be updated after moving to a new hypothesis $H'$. The fact that we can do this faster than creating the array from scratch is the key to JLE's acceleration.




\parheading{JLE formalization:} 
First the algorithm needs to compute the array $\Delta_{H_0}$ for the first iteration of Greedy. We omit these details due to space; see function ComputeInitalDelta in 
\ifEurosys Algorithm 2 of ~\cite{FlockFullDraft} \else Algorithm~\ref{HypothesisSearchAlgorithm} in appendix \fi which follows the intuition above.

Next, we describe how the $\Delta_H$ array can be updated when the greedy algorithm moves to a new hypothesis $H'$. We first note that $LL(H)$ is a sum of contributions from all flows and can be written as $LL(H) = 
\sum_{F\in \text{flows}} LL_F(H)$, where $LL_F(H) = \log P[F|H]$. We have 
\begin{align*}
\Delta_H(l) &= LL(H\oplus l) - LL(H)  \\
=& \sum_{F\in \text{flows}} LL_F(H\oplus l) - \sum_{F\in \text{flows}} LL_F(H) 
= \sum_{F\in \text{flows}} \Delta_H(l, F)
\end{align*}
where $\Delta_H(l, F) = LL_F(H\oplus l) - LL_F(H)$. Next we derive a useful property about the individual flow contributions $\Delta_H(l, F)$.

\vskip -0.15in
\begin{definition}
A flow $F$ intersects with link $l$ if at least one of the possible paths for $F$ has link $l$.
\end{definition}
\vskip -0.15in
\begin{theorem}\label{hypothesis search acceleration logic}
For a link $l'$ and hypothesis $H$, let $H' = H \oplus l'$.  Then for all links $l$ and flows $F$,
\begin{enumerate}[leftmargin=*]
    \item[(i)] If $F$ does not intersect with $l$, then $\Delta_H(l, F) = 0$
    \item[(ii)] If $F$ does not intersect with $l'$, then $\Delta_{H'}(l,F)\!=\Delta_H(l,F)$.
\end{enumerate}
\vskip -0.15in
\end{theorem}
\begin{proof}
This can be easily seen by expanding $\Delta_H'(l, F)$ and $\Delta_H(l, F)$. 
We note that for any $H$, the log likelihood of a flow $F$, given by $LL_F(H) = \log P[F|H]$, does not depend on a link $l$ if $F$ does not intersect with $l$ (that is, $LL_F(H) = LL_F(H\oplus l)$).

For (i),  if link $l$ does not intersect with $F$, then flipping $l$'s status does not affect $F$: $LL_F(H)\! = LL_F(H\oplus l) \Rightarrow \Delta_H(l,F)=0$.

For (ii), when $l'$ does not intersect with flow $F$, $LL_F(H') = LL_F(H\oplus l') = LL_F(H)$ and $LL_F(H'\oplus l) = LL_F(H\oplus l' \oplus l) = LL_F(H\oplus l)$. Consequently, we get $\Delta_{H'}(l, F)\! = \Delta_H(l, F)$.  \end{proof}

Hence, to obtain $\Delta_{H'}(l)$ from $\Delta_{H}(l)$, we only need to update the terms $\Delta_{H}(l, F)$ for flows $F$ that intersect with \emph{both} links $l'$ and $l$ since all other flow contributions to $\Delta_{H'}(l)$ remain unchanged from $\Delta_{H}(l)$. Let flows$(l', l)$ denote the set of flows that intersect with both $l$ and $l'$.
After updating the current hypothesis from $H$ to $H'$,  we can compute the new entry $\Delta_{H'}(l)$ for link $l$ using Theorem~\ref{hypothesis search acceleration logic}: 

\begin{align*}\label{DeltaUpdateEquation}
 &\Delta_{H'}(l)=\!\sum_{\mathclap{F\in \text{flows}}}\!\Delta_{H'}(l, F) = \sum_{\mathclap{F\in \text{flows}}} \Delta_H(l, F) + \sum_{\mathclap{F\in \text{flows}(l', l)}} \Delta_{H'}(l, F) - \Delta_H(l, F) 
 \end{align*}

\begin{equation} 
\Rightarrow \boxed{\Delta_{H'}(l)  =  \Delta_H(l) + \sum_{\mathclap{F\in \text{flows}(l', l)}}\! \Delta_{H'}(l, F) - \Delta_H(l, F)}
\end{equation} 


\noindent Once we have equation~\ref{DeltaUpdateEquation}, the algorithm to update the $\Delta$ array, for all $n$ entries, is simple to state. After moving to $H' = H\oplus l'$, we iterate over all flows that intersect with $l'$. For each such flow $F$, let $L_F$ be the set of links that intersect with $F$. For each $l\in L_F$, we update $F$'s contribution to $\Delta_{H'}(l)$. 
With memoization, one can update all entries $\Delta_{H'}(l, F)$ for all $l \in L_F$ in a couple of passes over $L_F$, similar to how we initially computed $\Delta_{H_0}$. The crux of the greedy+JLE algorithm is outlined in \ifEurosys Algorithm 1 of ~\cite{FlockFullDraft} \else  Algorithm~\ref{JLEAlgorithmShort} in appendix \fi and the full pseudocode in \ifEurosys Algorithm 2 of ~\cite{FlockFullDraft} \else Algorithm~\ref{HypothesisSearchAlgorithm} is outlined in appendix \fi . 


Given $LL(H)$, an alternate approach is to compute $LL(H\oplus l)$ without JLE, individually for each $l$, as in ~\cite{sherlock,netsonar}. This requires updating the contribution of all flows that intersect with $l$ since their likelihoods $LL_F(H\oplus l)$ would have changed after flipping the status of link $l$. Thus, the number of flows whose contributions need to be updated for computing just one entry $LL(H\oplus l)$ for a single $l$ is the same as that for computing all $n$ entries of the $\Delta$ array jointly with JLE. Thus, JLE results in in a $O(n)$ speedup. The reason for this large improvement is that JLE tracks the change in the $\Delta$'s (i.e., the difference in the differences:  $LL_F(H\oplus l) - LL_F(H)$) across iterations which allows reuse of computation from the previous iteration.  



Besides greedy search, JLE can apply to any algorithm which explores a hypothesis $H$'s neighbors: $H\oplus l$ for all $l$. This includes brute force, Sherlock and NetSonar's inference, and MCMC techniques (e.g. Gibbs sampling). 
Using JLE, we were able to accelerate (i) Sherlock's inference (\ifEurosys Algorithm 3 in~\cite{FlockFullDraft} \else Alg.~\ref{JLEBruteForceAlgorithm} in appendix\fi), and (ii) Gibbs sampling for Flock, both by multiple orders of magnitude. We ended up using Greedy for \Sys because (i) Sherlock's inference can not detect $K>2$ concurrent failures and was still slow with JLE (\S~\ref{runningTimeSection}) and (ii) for Gibbs sampling, it's hard to bound the number of iterations required for convergence.
Gibbs sampling~\cite{networkTomographyGibbsSampling} without JLE was too slow for our purposes.





\section{\Sys: Analysis} \label{flockAnalysis}


\subsection{Runtime analysis} \label{runtimeAnalysis}
Let $n$ be the number of links, $m$ be the number of flows, $T$ be an upper bound on the number of links that any flow intersects with, $D$ be an upper bound on the number of flows that any link intersects with and $K$ be the maximum number of concurrent failures (note \Sys's inference does not know $K$). The runtime of Greedy inference with JLE is $O(n+mT+(K\!-\!1)DT)$. If we had used just Greedy without JLE (computing likelihood of each hypothesis individually), the runtime would be $O(n + mT + (K\!-\!1)nDT)$. In contrast, Sherlock's runtime is $O(n^KDT)$. JLE can improve Sherlock's runtime by a factor of $n$, to $O(n^{K-1}DT)$. 
From the analysis above (and our experiments later), it can be seen that Greedy + JLE is dramatically faster. See \ifEurosys section C of \cite{FlockFullDraft} \else \S~\ref{fullRuntimeAnalysis} of appendix \fi for derivations of these results.

\subsection{Accuracy analysis}
\label{greedyJustification}
\noindent 

We now analyze conditions in which Greedy returns the true MLE hypothesis.
To make the problem tractable, we restrict the analysis to cases where path taken is known (true for active probes and INT) and packets crossing a link get dropped independently according to a (unknown) drop probability of that link (inference is NP hard if packets get dropped adversarially, see \S~\ref{adversarialVersionNPHardAppendix} in appendix).
Theorem~\ref{greedyGuaranteeConditions} gives a sufficient condition on the traffic pattern for \Sys's inference to correctly recover the set of failed links, providing intuition for why Flock's model and inference work well in practice. 
\vskip -0.1 in
\begin{definition}
For given traffic $T$, let $T(\{l_1,l_2,...l_k\})$ denote the number of packets that each go through all of the links $\{l_1, l_2, ... l_k\}$. If $T(\{l_1, l_2\})/T(\{l_1\}) \leq \epsilon$ for all links $l_1$ and $l_2$, then we say $T$ is \textit{$\epsilon$-skewed}.
\end{definition}
\vskip -0.1 in
\noindent \begin{theorem} \label{greedyGuaranteeConditions} 
For any topology with $(1/\alpha)$-skewed traffic, with high probability, \Sys's inference returns the set of all failed links if the number of failures is $\leq \alpha/2$, the number of packets $T_{min}$ crossing every link is larger than a certain threshold, and the drop probabilities are $<p_g$ on all good links and $>p_b$ on all failed links where ($p_g$, $p_b$) satisfy the condition $5p_g < p_b < 0.05$.
\end{theorem}
\vskip -0.05in
\noindent Proof in \ifEurosys ~\cite{FlockFullDraft} \else appendix\fi. Theorem~\ref{greedyGuaranteeConditions} has an intuitive interpretation: Consider two links $l_1$ and $l_2$, where $l_1$ has failed and $l_2$ works correctly. Intuitively, at most $\frac{1}{\alpha}$ fraction of the dropped packets on $l_1$ will transit $l_2$. One can think of this as $l_2$ getting ``$\frac{1}{\alpha}$ fraction of the blame'' for the dropped packets on $l_1$. If there are $(f\alpha)$ failures and these are arranged adversarially so that all of them add blame to $l_2$, then in the worst case $l_2$ can get $(f\alpha)\cdot \frac{1}{\alpha} = f$ times as much blame as a true failed link.  Intuitively, a large enough constant $f$ would confuse the algorithm into classifying $l_2$ as failed. The theorem proves that for $f < 0.5$, the algorithm outputs the true failed links.




\section{Implementation} \label{implementation}


\subsection{Agent and inference engine}

\noindent We implement an agent that runs on end-hosts and collects flow statistics via a lightweight packet dumping tool based on the PF\_RING module~\cite{pfring}. 
The agent periodically encapsulates the collected flow statistics (52 bytes per flow) into export IPFIX messages, and sends it to the collector. 
As the specifics of the agent/collector don't affect our results, we leave more optimized designs for future work. Note that commercial solutions also exist~\cite{solarwinds,manageengine}, possibly employing alternate approaches such as pulling flow statistics directly from the kernel via eBPF or using multiple collectors at scale.

Flock's inference engine, written in C++, (i) collects IPFIX flow reports from  agents and (ii) periodically runs inference on the collected input. 
Currently, the network topology is static, but the inference engine can be modified to obtain the topology from a controller.  
The engine periodically reads flow reports from the queue, every 30 seconds,  and runs the inference algorithm of \S~\ref{sec:inf-alg}.

\subsection{Parameter Calibration}  \label{sec:paramCalibration}
A important problem we encountered, with both \Sys and the other competing schemes, is how to set their hyperparameters.  \Sys has 3 hyperparameters ($p_g$, $p_b$, $\rho$), NetBouncer has 3, and 007 has 1. In real deployed systems, parameters and thresholds are quite common, and are set based on past deployment experience.  However, manual tuning puts extra onus on the user 
and makes it difficult to evaluate systems. 
The manually set parameters of past schemes 007~\cite{007} and NetBouncer~\cite{netbouncer} gave suboptimal results in our environments, across different topology scales and failure scenarios.

Thus, we design an automated parameter calibration method that we use for all schemes.
We use a training set of monitoring data to search for the parameter settings that obtain the best precision and recall in the training set.  It is tempting to obtain the training set from historical monitoring data, which is generally available in deployed systems.  However, this can be tricky, since: (a) historical data may not have faults labeled; and (b) faults, especially rare faults, may have diverse and unpredictable types so that historical data is not entirely representative of the next upcoming incident.


To solve (a), we leverage simulations to obtain the training set. We calibrate parameters once using this training set and use those parameters across our experiments, unless stated otherwise. Note that if this method were used in a deployment, it would increase concerns with (b) since simulations may not match the real world. 
To address (b) we will experimentally quantify the robustness of each scheme to scenarios where the train and test sets are drawn from \emph{different environments} (different topology, monitoring duration, fault rate, fault type).


Once we have the training set, we use the following calibration method.  For each hyperparameter, we choose equally-spaced values in a reasonable range of possible values. We fix a minimum precision $P$ and find the parameters which, in a training set, yielded highest recall and had precision $>P$.  Varying $P$ produces a set of parameters that are Pareto-optimal along the precision/recall tradeoff curve. Our evaluation will apply these parameters to a separate test data set.

To choose a single parameter setting (rather than a tradeoff curve), we set $P=98\%$ and find the setting that maximizes recall (in the training set); if no such point exists or recall is too low ($<25\%$), then we subtract $5\%$ from $P$ and try again, repeating until a setting is found. This method lays more emphasis on precision, which is usually desirable.


\section{Evaluation Methodology}\label{experimentalSetup}

\subsection{Systems evaluated}
\noindent We compare \Sys with several state-of-the-art datacenter fault localization schemes. Our codebase consists of 7K LOC  including C++ implementations of all inference algorithms.

We implemented the ``Ferret'' inference algorithm of \textbf{Sherlock} (Sec. 3.2 of~\cite{sherlock}).  
For a fair comparison, we run Ferret on the same PGM as Flock (which is anyway similar to Sherlock's; \S~\ref{inference graph model}).
Sherlock can not detect $K>2$ failures, but (as expected) resulted in the same accuracy as Flock for $K\leq2$ failures at small scale. Hence, we only show performance differences. 




We implemented \textbf{NetBouncer's} algorithm (Figure 5 in ~\cite{netbouncer}) and \textbf{007} (Algorithm 1 in ~\cite{007})
We verified our 007 implementation by matching scores outputted by the publicly available 007 code. We were also able to reproduce Figure 10 in~\cite{007} with our implementation/setup.

For all schemes, we calibrate parameters once using simulations of random packet drops  and use those parameters by default, unless stated otherwise. In some cases, we also show results with parameters calibrated on that environment.



\subsection{Input telemetry types} \label{information kinds}
\noindent We use four different kinds of input for inference:
\begin{itemize}[leftmargin=10pt]
\itemsep0em
    \item \textbf{A1}: Active probes between end-hosts and the core switches with known paths, as designed for NetBouncer~\cite{netbouncer}.

    \item \textbf{A2}: Reports about flows with $\geq 1$ retransmission, along with their (actively-probed) paths, as designed for 007~\cite{007}.

    \item \textbf{P}: Passive information consisting of reports about regular (application) data flows, whose traffic matrix is thus dictated by the network environment (\S~\ref{ns3Setup}). A set of possible paths is known (based on ECMP multipath).

    \item \textbf{INT}: We assume INT~\cite{int} provides reports, including paths, for both A1 and P (which thus becomes a superset of A2).

\end{itemize}

\noindent Note the last case is intended to test full deployment of INT; in other deployment modes it could trace just a subset of traffic. Similar reports could be obtained from other recent packet marking~\cite{FbLocalization} or mirroring~\cite{everflow,pathdump, sonata,omnimon,switchpointer} methods.
Since \Sys can incorporate different kinds of inputs, we compare accuracy across input types: \Sys (A1) vs NetBouncer (A1), \Sys (INT) vs NetBouncer (INT) and \Sys (A2) vs 007 (A2).  We also quantify the accuracy boost \Sys obtains from additional passive information (A1+P, A2+P, A1+A2+P).  NetBouncer and 007 cannot trivially ingest the passive telemetry as they do not model path uncertainty. Finally, the passive flow telemetry can be downsampled in a large datacenter with high link speeds to reduce volume of the monitoring data.

\ifshowpassive
\subsubsection{Reduced Analysis: passive data (P)} \label{reduced analysis}
For this section, passive information (P) includes all application flows irrespective of how many packets were dropped (excluding active probes for NetBouncer). We assume that we don't know the path for any flow in the input. Although this dataset is unbiased (since all application flows are included), making use of it is non-trivial for the following reason. In a Clos topology, there are sets of ``symmetric" links. For e.g. in a 3-tiered Fat-tree, a flow whose destination is in a different rack can take any of the ToR uplinks in the case of ECMP routing. Hence, just knowing the set of possible paths for all flows is not enough to differentiate between the ToR uplinks, since all outgoing flows can each get routed through any of the uplinks. We say that the uplinks of any given ToR form an equivalence class. No inference algorithm can differentiate between the links of an equivalence class based only on passive flow information with only the path set known for each flow. This is also illustrated by Fig.~\ref{fig:eq_classes}.

Hence, the best our inference algorithm can do is to localize the failure to within an equivalence class. To localize the equivalence classes, one can think of inference in a reduced graph, obtained by collapsing all links in the same equivalence class. The hypotheses space for the reduced topology differentiates between 0 or $\geq1$ failures in an equivalence class. Note that there is only one possible path in the reduced graph for any flow. The probability expression for $P(f|H)$ is same as the original model, the difference being that the number of failed paths in the path set is determined by which equivalence class has failed as per the hypothesis $H$. For e.g., if the host-ToR uplink is labeled failed by $H$, then all paths in the pathset for $f$ are deemed to have failed. Alternatively, if $H$ classifies an equivalence class consisting of links from aggregation layer of one pod to the core switches, then we classify one path as failed in the entire pathset. Note that failure of an equivalence class might mean one link in the class has failed or more than one links have failed. Conservatively we assume that only one link has failed, which we found, works well in practice.  

\fi



\subsection{Network environments}\label{ns3Setup}
\parheading{NS3 simulations.} We set up a NS3 simulation to output a trace consisting of flow metrics (retransmissions/packets sent). We feed this trace as input to inference. We use a standard 3-tiered Clos topology~\cite{AlFares08} with 2500 40Gbps links, ECMP routing and 3x oversubscription at ToRs. Like ~\cite{netbouncer}, we set drop rates on all non-failed links between $0-0.01\%$ chosen independently and uniformly at random to model occasional drops on good links~\footnote{TCP can tolerate such low rates, hence it's reasonable to qualify such links as non-faulty. These rates are not central to \Sys}. 
For all our experiments, half the traces used uniform random traffic and the other half used a skewed traffic pattern where $50\%$ of the traffic is concentrated among $5\%$ of the racks, randomly chosen. Flow sizes were drawn from a Pareto distribution (mean: 200KB, scale:1.05) to mimic irregular flow sizes in a typical datacenter~\cite{hull}.  

\parheading{Large scale simulation.} NS3 was too slow for large scale simulations. Hence, we use a flow level simulator (similar to~\cite{007}), that drops each packet as per preset drop probabilities on links but does not model queuing or TCP. We use this simulator for scaling experiments (\S~\ref{runningTimeSection}). 


\parheading{Hardware test cluster.} \label{testbedSetup}
We set up a physical testbed with 10 switches and 48 emulated hosts, each with its own dedicated hardware NIC port and one CPU core. One of the 48 hosts runs Flock's collector. 
We use a standard 2-tier Clos topology with 2 spines, 8 leaf racks and 6 hosts per rack. 
We provision 1 Gbps link speeds to emulate as many hosts as possible. 
Schemes with A1 are omitted from our testbed results since our switches don't have the in network IP-in-IP feature for A1~\cite{netbouncer}.


\subsection{Failure scenarios} \label{failureScenarios}

\parheading{In simulation:}
\begin{itemize}[leftmargin=*]
\item \textbf{Silent link packet drops:} a link drops a small fraction of packets without updating switch counters. Silent drops are a common problem in the industry~\cite{FbLocalization,007,netbouncer}. 

\item \textbf{Silent device failure:} An error in a device component (e.g., memory, line card) causes silent packet drops. This differs from the prior scenario because it affects many or all links on the device.   
\end{itemize}

\parheading{In hardware test cluster:} 
\begin{itemize}[leftmargin=*]
    \item \textbf{Queue misconfiguration}: A WRED queue drops packets with probability $p$ when the queue length is above a configurable threshold $w$. 
    We misconfigure WRED queues~\cite{Red} on switches, setting $p=1\%$ (available choices: 1-100\%) and $w=0$ (so, the link works normally if the queue is empty).
    \item \textbf{Link flap}: We pull out a cable manually and quickly put it back in to emulate link flaps~\cite{CiscoLinkFlap}. In our setup, link flaps caused the latency of the flows transiting the link to spike, but did not produce any significant increase in retransmissions (i.e., the link was buffering packets). 
\end{itemize}

\begin{figure*}
    \centering
    \begin{minipage}[b]{0.13\textwidth}
    \centering
        \begin{subfigure}[b]{\linewidth}
        \includegraphics[width=1.05\textwidth]{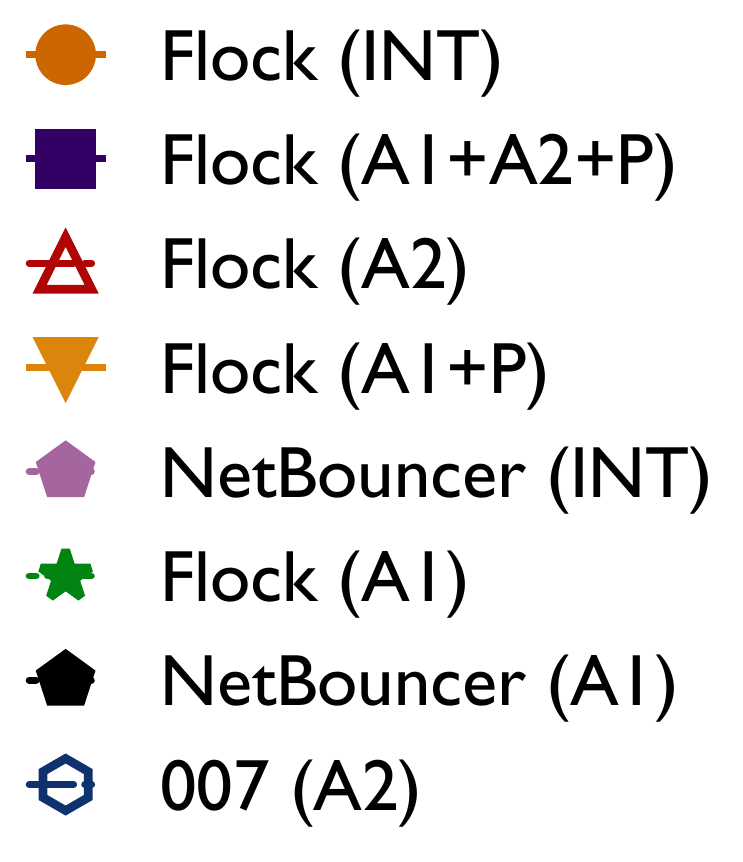}
        \vspace{0.0in}
        \label{fig:mixed_legend}
        \end{subfigure}
    \end{minipage}
    \begin{minipage}[b]{0.215\textwidth}
        \begin{subfigure}[b]{\linewidth}
            \includegraphics[width=1.12\textwidth]{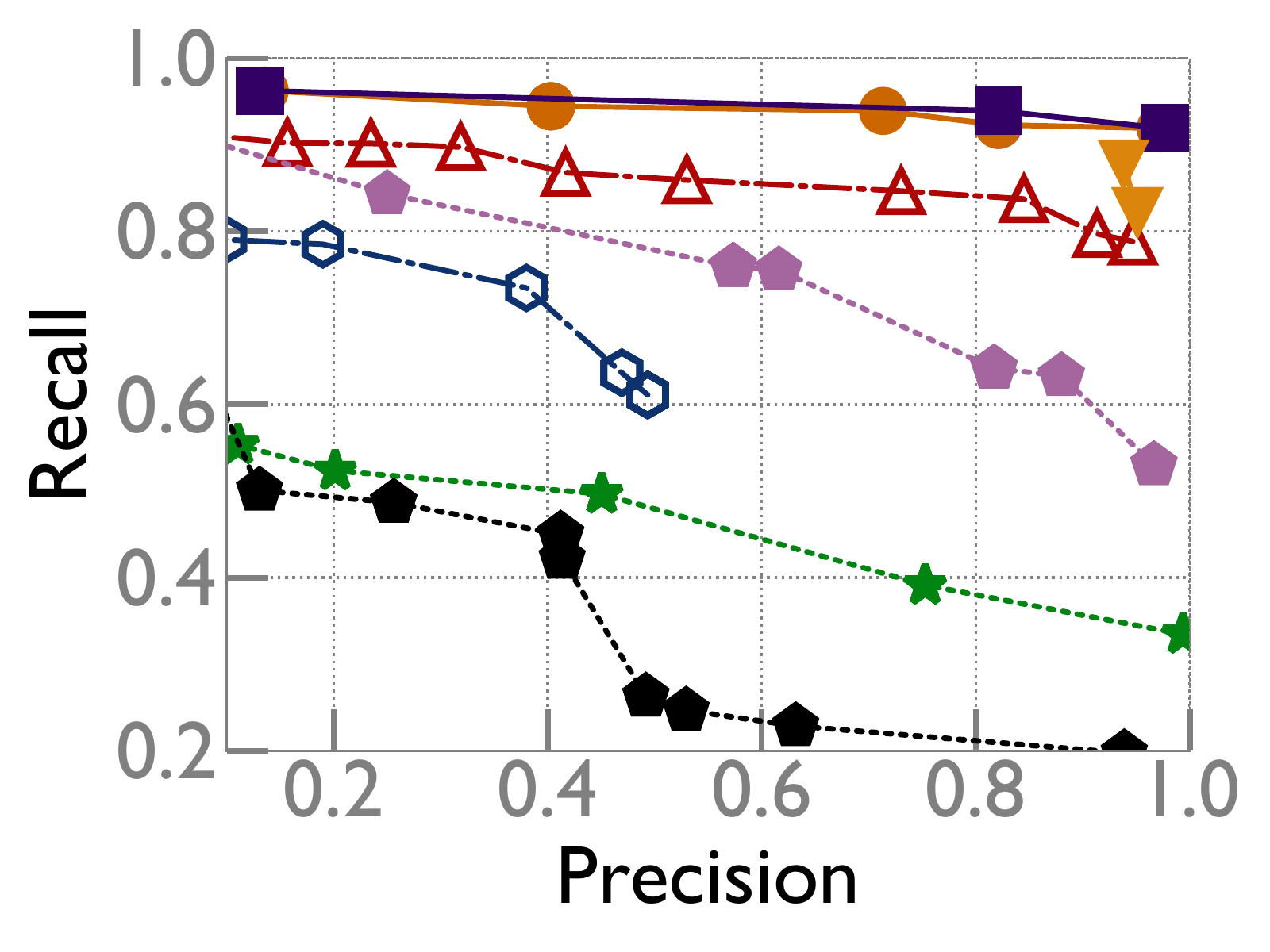}
            \caption{With 100K flows}
            \label{fig:mixed_tdot25}
        \end{subfigure}
    \end{minipage}
    \hskip 0.5cm
    \begin{minipage}[b]{0.215\textwidth}
    \centering
        \begin{subfigure}[b]{\linewidth}
        \includegraphics[width=1.12\textwidth]{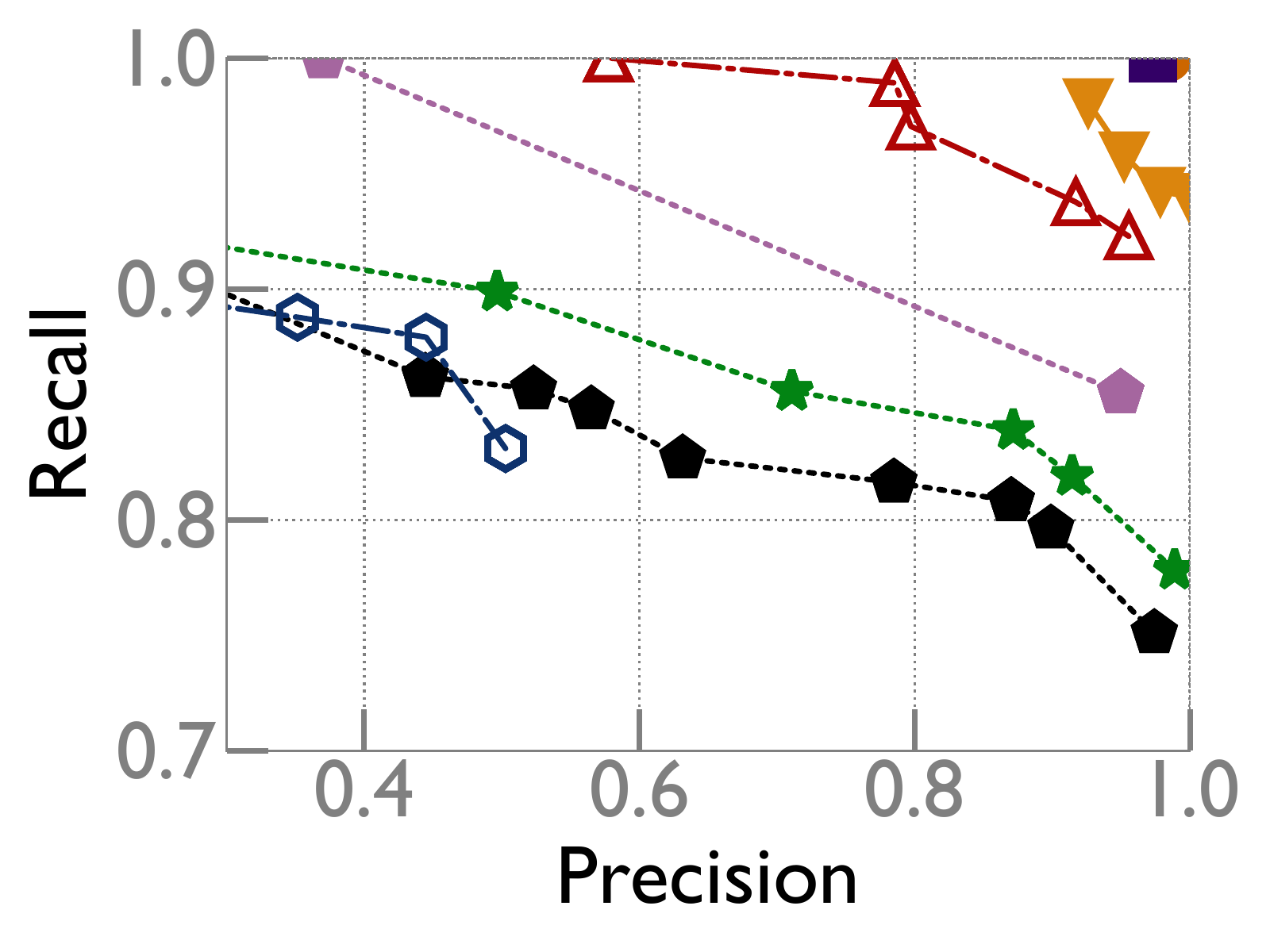}
        \caption{With 400K flows}
        \label{fig:mixed_t1}
        \end{subfigure}
    \end{minipage}
        \hskip 0.3cm
    \centering
    \begin{minipage}[b]{0.215\textwidth}
        \centering
        \begin{subfigure}[b]{\linewidth}
            \includegraphics[width=1.08\textwidth]{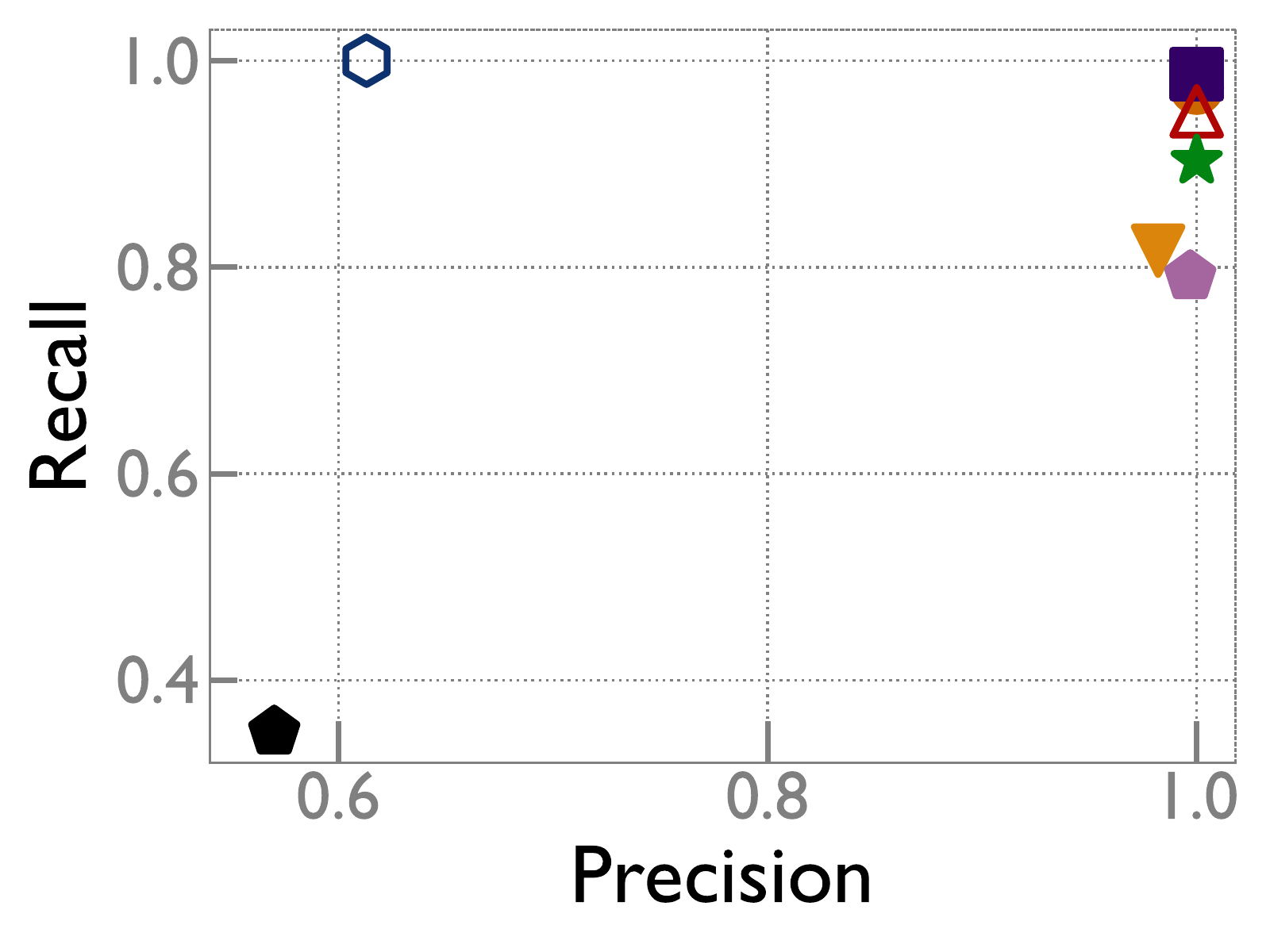}
            \caption{Device failures}
            \label{fig:deviceFailure}
       \end{subfigure}
    \end{minipage}

     \caption{ Accuracy for silent packet drops. (a), (b): Tradeoff curves for NetBouncer, 007 and \SystemName for silent drops, varying hyperparameters for each scheme. Schemes are annotated with the input information they use. (c): Accuracy on device failures.  }
  \label{fig:mixed_traces}
\end{figure*}

\begin{figure}
    \hskip -0.6 cm
    \begin{minipage}[b]{0.22\textwidth}
        \centering
        \begin{subfigure}[b]{\linewidth}
            \includegraphics[width=1.2\textwidth]{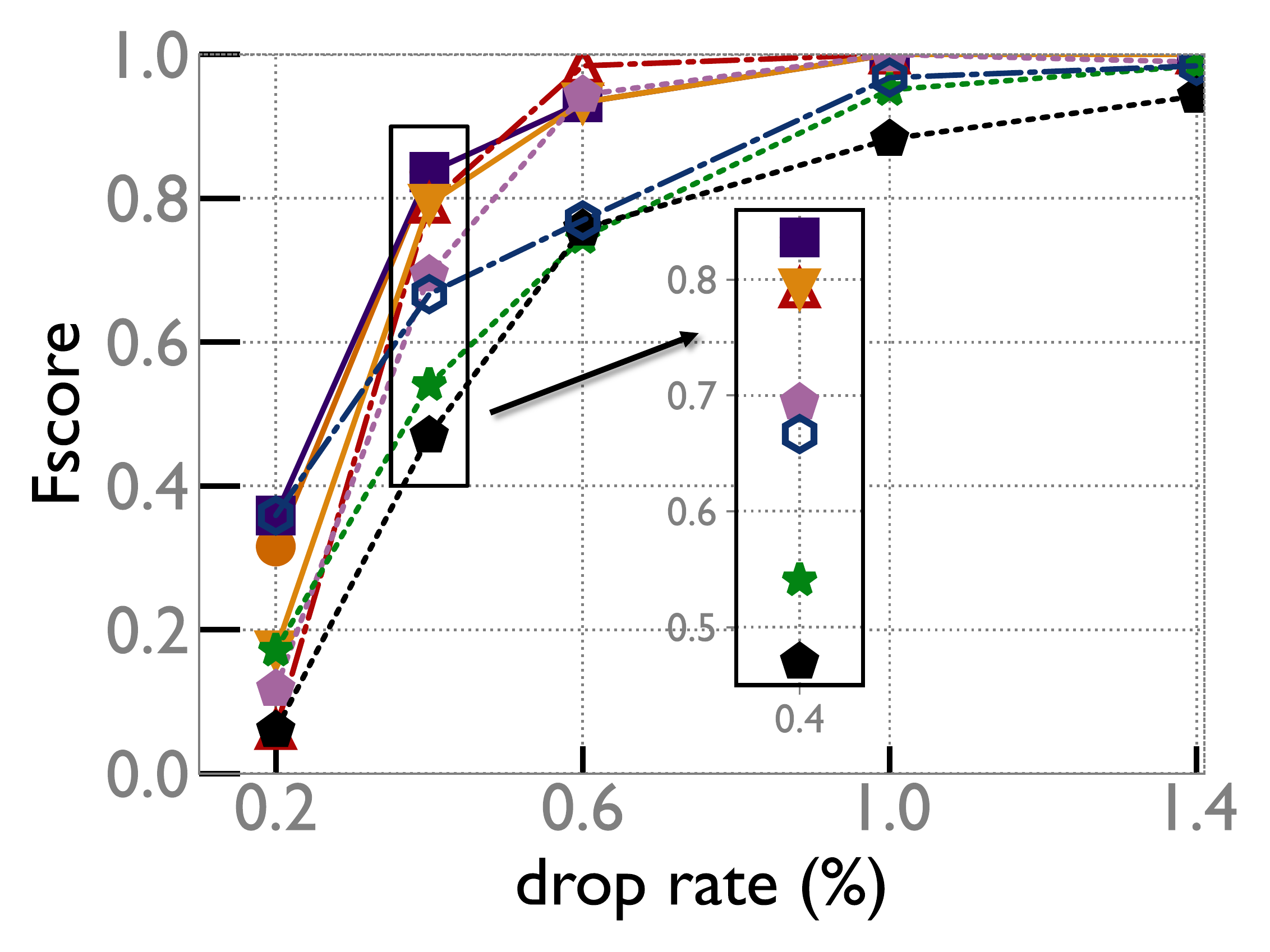}
            \caption{Uniform traffic}
            \label{fig:softnessRandom}
        \end{subfigure}
    \end{minipage}
    \hskip 0.7 cm
    \begin{minipage}[b]{0.22\textwidth}
        \centering
        \begin{subfigure}[b]{\linewidth}
            \includegraphics[width=1.2\textwidth]{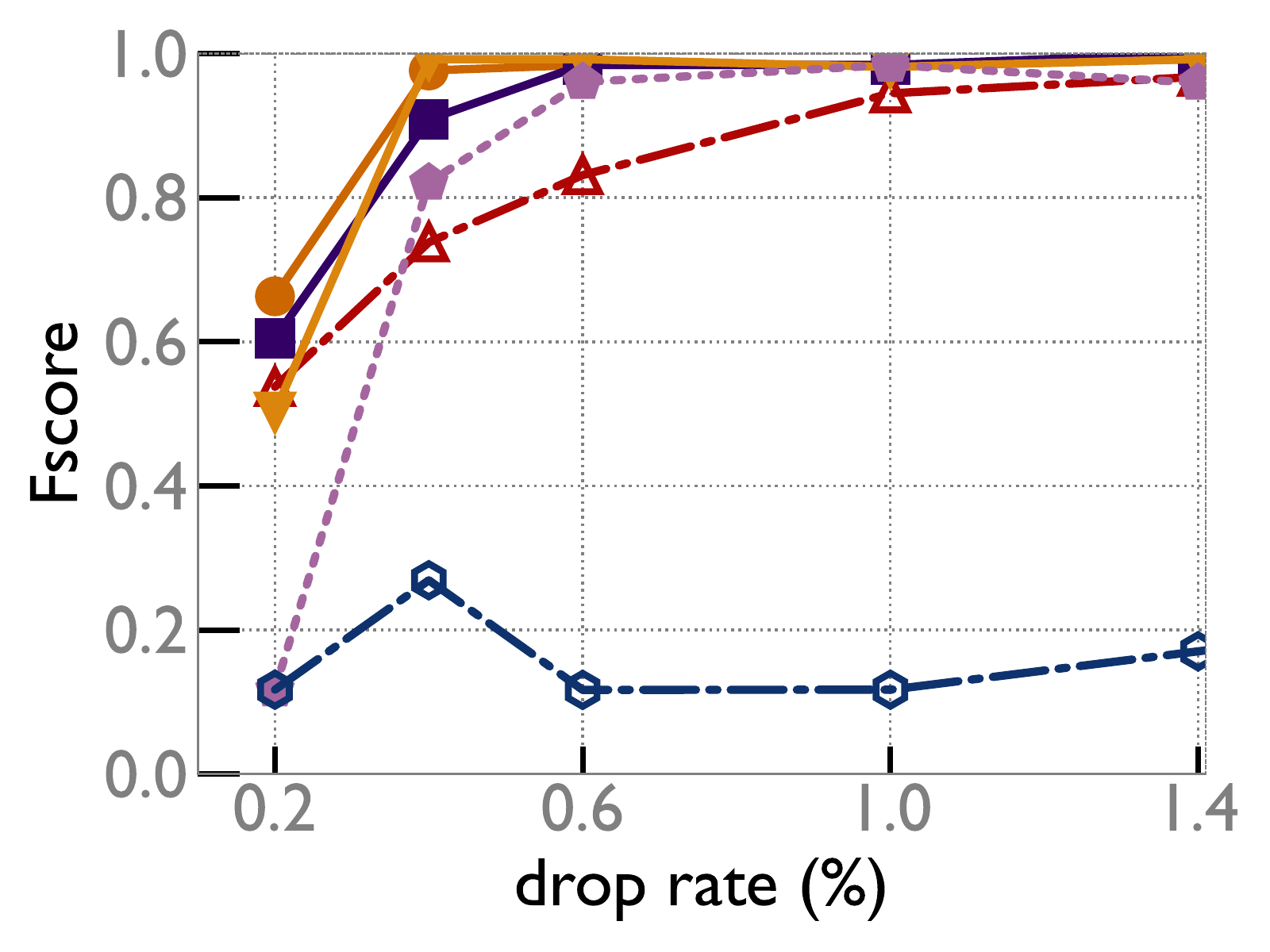}
            \caption{Skewed traffic}
            \label{fig:softnessSkewed}
        \end{subfigure}
    \end{minipage}
  \caption{The extent of drop rates that each scheme can detect. 
  }
  \label{fig:SoftnessDevice}
\end{figure}

\parheading{Evaluation metrics}: We use \textbf{Precision} (fraction of predicted failed links/devices that actually failed) and \textbf{Recall} (fraction of failed links/devices that were correctly predicted as failed), to quantify false positives and false negatives  respectively. 
We use the standard \textbf{Fscore} measure (harmonic mean of precision and recall) when we need a combined measure of accuracy.

\section{Evaluation Results} \label{experimentalResults}

\noindent The first goal of our evaluation is to investigate Flock's accuracy compared to NetBouncer and 007. 
We compare various input types (INT, A1, A2, P) to quantify benefits obtained from incorporating passive data and INT.  As expected, Flock and Sherlock had the same accuracy in small scale experiments (with max $K=2$ failures), where Sherlock finished in reasonable time.
Hence we don't show accuracy comparisons with Sherlock.

Next, we investigate performance of \Sys, including inference algorithm compared to Sherlock, NetBouncer and 007 and the runtime benefits of using JLE to speed up inference. 


\ifshowpassive
\begin{figure*}
  \begin{minipage}[b]{0.32\textwidth}
    \centering
    \hskip -1.6 cm
  \begin{subfigure}[b]{\linewidth}
    \includegraphics[width=0.95\textwidth]{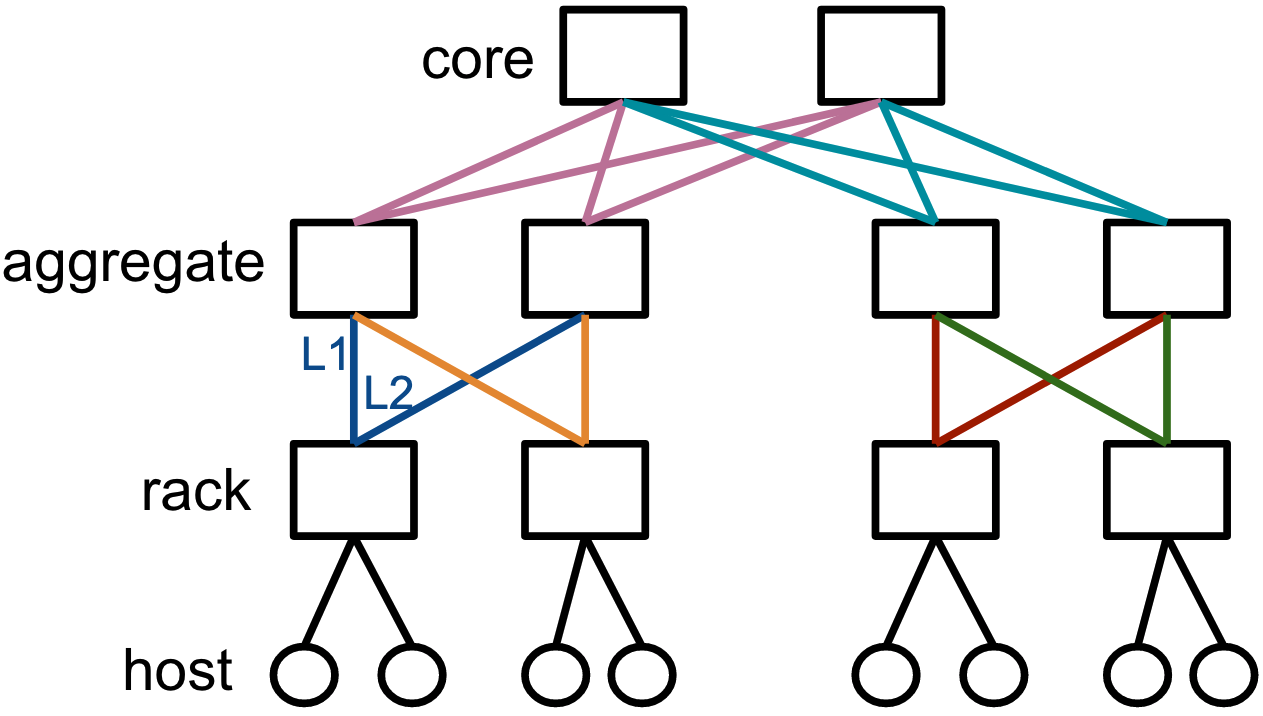}
    \scompactcaption{Equivalence classes of links shown using color. With passive data, one can localize a failure to only within an equivalence class. For e.g. it is not possible to differentiate between L1 and L2 without path information of flows.}
    \label{fig:eq_classes}
  \end{subfigure}
  \end{minipage}
    \begin{minipage}[b]{0.30\textwidth}
    \centering
  \begin{subfigure}[b]{\linewidth}
    \includegraphics[width=0.8\textwidth]{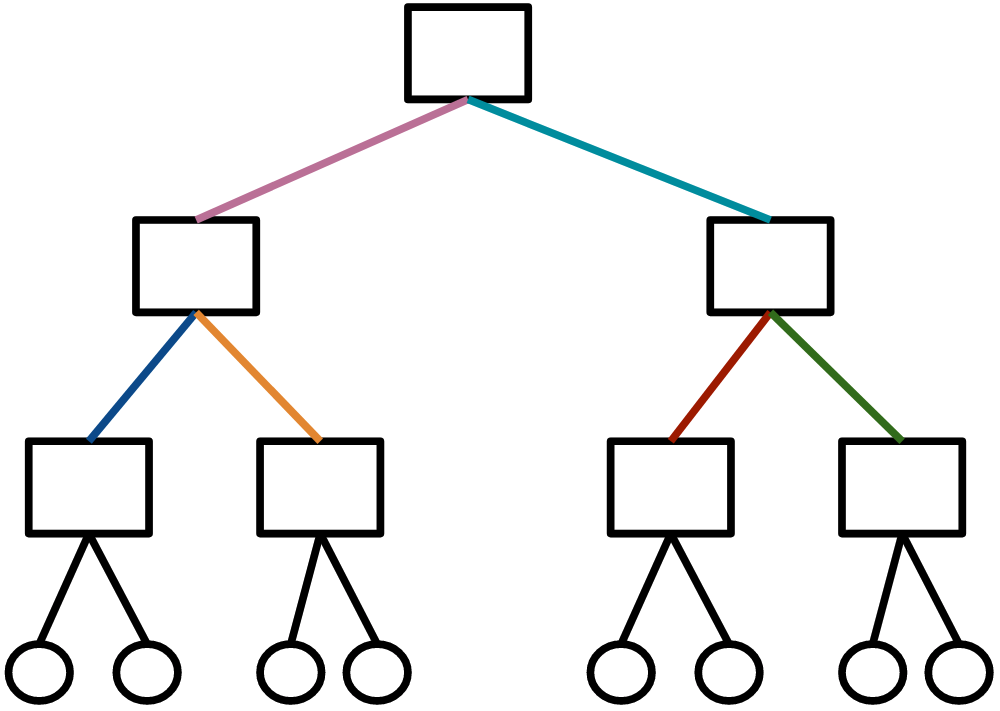}
    \scompactcaption{Reduced topology by collapsing all links in the same equivalence class.}
    \vskip 1cm
    \label{fig:eq_classes}
  \end{subfigure}
  \end{minipage}
   \begin{minipage}[b]{0.35\textwidth}
    \centering
    \begin{subfigure}[b]{\linewidth}
    \includegraphics[width=1.0\textwidth]{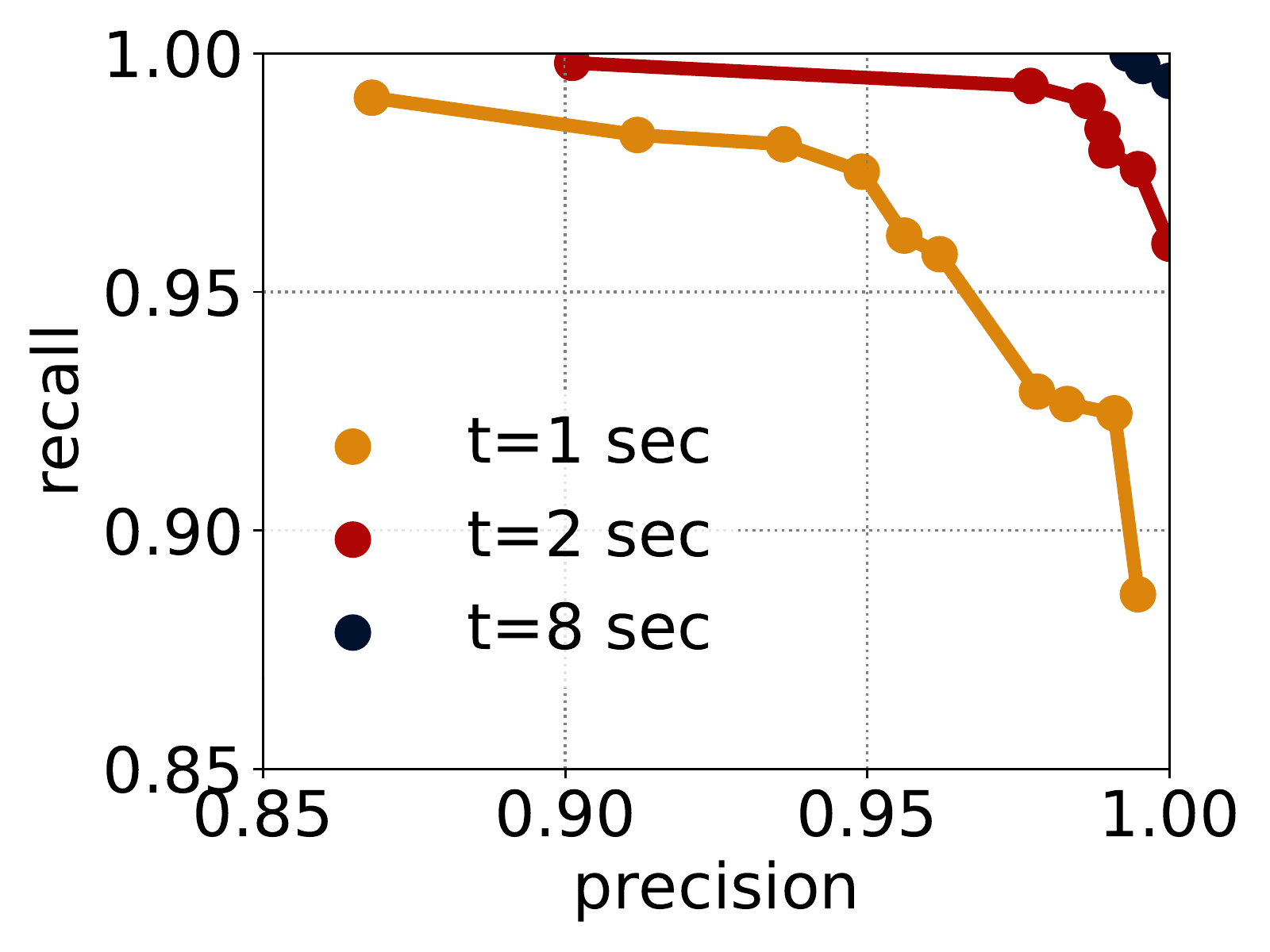}
    \compactcaption{\textcolor{red}{redo} Precision and recall for inference with only passive data. Precision/recall is defined w.r.t. equivalence classes.}
    \label{fig:pr_passive}
  \end{subfigure}
  \end{minipage}
\end{figure*}
\fi

\subsection{Silent packet drops} \label{mixedTracesExperiments}
\noindent We generated 63 traces via NS3, each with 1 to 8 failed links with drop rate on each failed link chosen uniformly at random between 0.1\% and 1\%~\cite{netbouncer} 
(drops on good links are set as in \S~\ref{ns3Setup}). For each trace, we send active flows between the hosts and the core switches (A1), sending 40 packets per second and 400K passive flows per second across all hosts. 
Fig.~\ref{fig:mixed_traces} shows precision-recall tradeoff curves, with different calibrated parameters (\S\ref{sec:paramCalibration}) after 100K and 400K flows. 
Using the chosen point for each scheme (\S~\ref{sec:paramCalibration}), we draw the following conclusions for each input type:
\begin{itemize}[leftmargin=7.0pt]
\itemsep0em
    \item \textbf{A1}: Flock reduces error rate over NetBouncer  by roughly 45\% (fscore: 0.5 vs 0.27). With  4$\times$ more probes, Flock still reduces the error rate by >20\% compared to NetBouncer (fscore: 0.87 vs 0.84).
    \item \textbf{A2}: Flock (fscore: 0.93) reduces error-rate over 007 (fscore: 0.61) by 5.5x after 400K flows. 
    \item \textbf{A1+P and A1+A2+P}: When active probes (A1, A2) are augmented with passive information (P), \SystemName achieves very high accuracy (fscore with A1+A2+P: 0.98, A1+P: 0.93) after 400K flows (see Fig.~\ref{fig:mixed_tdot25}), suggesting that these schemes require less data than active-only schemes for localization.
    \item \textbf{INT}: Flock (INT) achieves the best accuracy (fscore 0.99) reducing error over NetBouncer (INT) (fscore 0.88) by 12x. 
\end{itemize}

\noindent The last two points highlight the benefits of incorporating passive data for accuracy. 
007's performance can be attributed to it being sensitive to traffic skew
(\S~\ref{softnessExperiments}).

\subsection{Device failures} \label{deviceFailureExperiments}
\noindent Using the same setup as \S~\ref{mixedTracesExperiments}, we simulate a device failure by failing $f\%$ of a faulty device's links. We generate 64 traces, each consisting of up to 2 device failures,  varying $f$ across traces from $25\%$ to $100\%$. A subset of links failing on a device is similar to the behaviour of a faulty line card on that device. 
For all schemes, we used the same parameters as in \S~\ref{mixedTracesExperiments} (we calibrated NetBouncer's threshold for the number of problematic flows crossing a device). As shown in Fig.~\ref{fig:deviceFailure}, Flock outperforms NetBouncer and 007 for all types of information. Flock (INT) achieves $\approx$100\% recall, 
compared to 80\% recall of NetBouncer (INT). Flock (A2) reduces error-rate compared to 007 by 8x (fscore 0.97 vs 0.76). Flock (A1+P) has poorer precision for device failures than link failures.

\subsection{Soft gray failures} \label{softnessExperiments}
\noindent  We vary the drop rate on a single failed link to test what drop rates Flock can detect. A useful metric is the ratio of drop-rate on a failed link and the maximum drop-rate on a functioning link, which we call Signal to Noise Ratio (SNR), by a slight abuse of terminology.
We use the same setup and parameters as \S~\ref{deviceFailureExperiments} (except 007 which had to be calibrated separately for skewed traffic since otherwise it had poor recall). We used 32 traces for each data point.
From Figs.~\ref{fig:softnessRandom} and~\ref{fig:softnessSkewed}, we conclude that \Sys can detect links with $>1\%$ drop rate (or SNR $>$ 100) with high recall, with A2. With uniform traffic, 007's accuracy is good when SNR > 100 (consistent with the SNR in Fig.~10 of~\cite{007}).
007's recall gets affected significantly with skewed traffic. 
After adding passive telemetry with either INT or (A1+A2+P), \Sys's accuracy gets boosted and it is able to detect $>0.4\%$ drop rate reliably. NetBouncer's accuracy is slightly worse than Flock for A1, but its accuracy becomes even worse for multiple concurrent failures with different drop rates (\S\ref{mixedTracesExperiments}). 
Schemes utilizing A1 (active probes) are unaffected by skew in the application traffic and hence omitted from Fig.~\ref{fig:softnessSkewed}.

\ifshowpassive
\subsection{Only passive information} \label{benefits of passive info}

Fig.~\ref{fig:pr_passive} shows the precision/recall with only passive information. Note that precision/recall is defined w.r.t. the equivalence classes as no algorithm can do any better, without additional data. 
\fi

\subsection{Misconfigured queue} \label{misconfiguredQueue}
\noindent We move now from simulated faults to our testbed, beginning with misconfigured queues. Flock achieves high accuracy with all information types (see Fig.~\ref{fig:testbed_wred}). Using the same parameters as in \S~\ref{mixedTracesExperiments} for all schemes, Flock (INT) had higher precision and recall than NetBouncer (INT) (16x less error in fscore), whereas Flock (A2) had higher precision than 007 (A2) (the solid markers in Fig.~\ref{fig:testbed_wred}). For comparison, we also show precision recall tradeoff curves with parameters calibrated on the testbed with real examples. In this case, Flock (INT) had 7x less error than NetBouncer (INT) (fscore: 0.87 vs. 0.98) and Flock (A2) had 16x lower error compared to 007 (fscore: 0.97 vs. 0.5) (the hollow markers in Fig.~\ref{fig:testbed_wred}). 
Flock (A2+P) gets very close to Flock (INT), consistent with \S\ref{mixedTracesExperiments}.

\begin{figure*}
   \begin{minipage}[b]{0.23\textwidth}
    \centering
      \hskip -.3 cm
    \begin{subfigure}[b]{\linewidth}
    \includegraphics[width=1.1\textwidth]{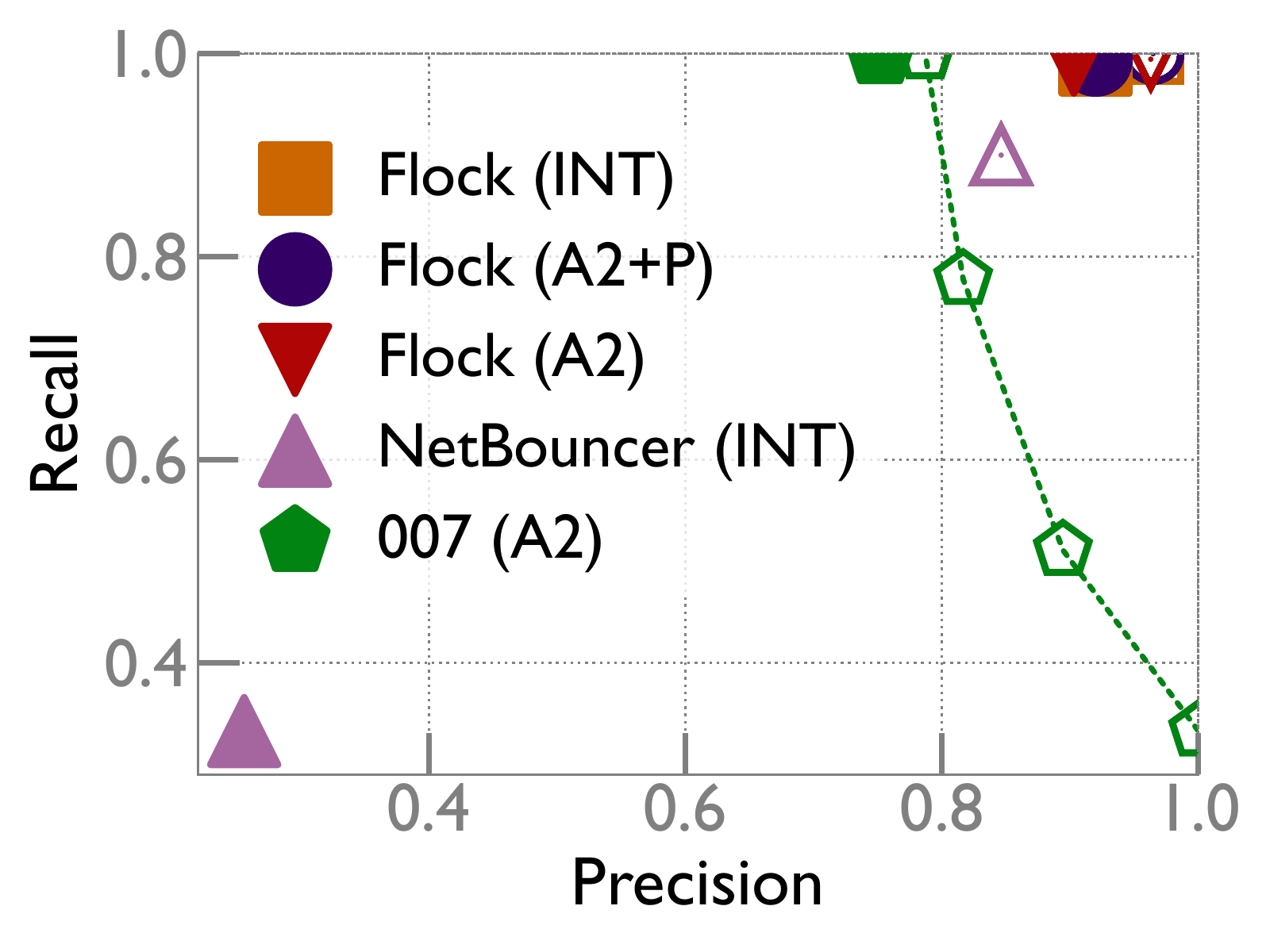}
    \caption{Misconfigured queue}
    \label{fig:testbed_wred}
  \end{subfigure}
  \end{minipage}
  \begin{minipage}[b]{0.23\textwidth}
    \centering
    \hskip .2 cm
    \begin{subfigure}[b]{\linewidth}
    \includegraphics[width=1.1\textwidth]{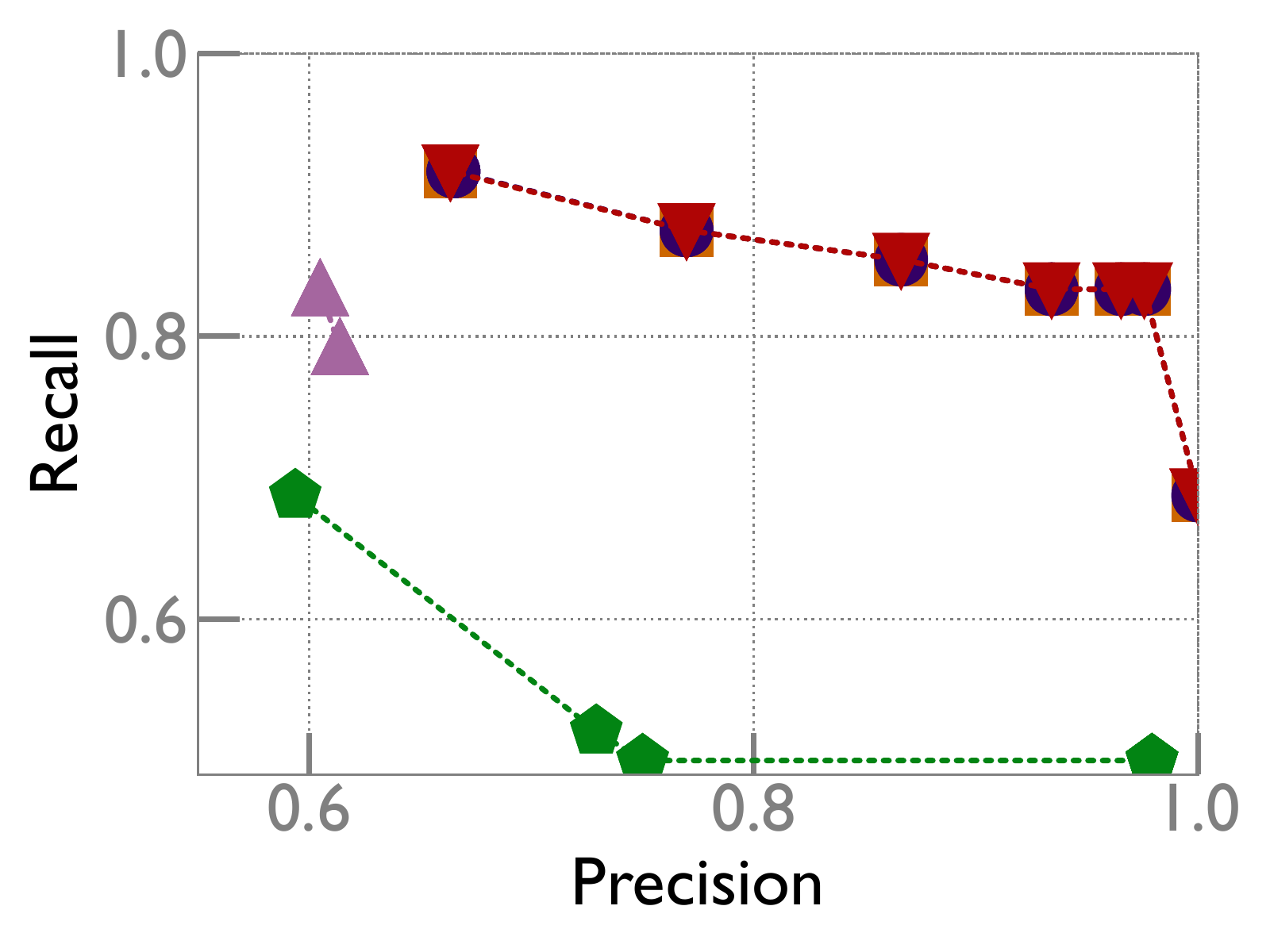}
    \caption{Link Flap}
    \label{fig:testbed_link_flap}
  \end{subfigure}
  \end{minipage}
     \hskip 0.2cm
    \begin{minipage}[b]{0.23\textwidth}
    \centering
    \begin{subfigure}[b]{\linewidth}
    \includegraphics[width=1.10\textwidth]{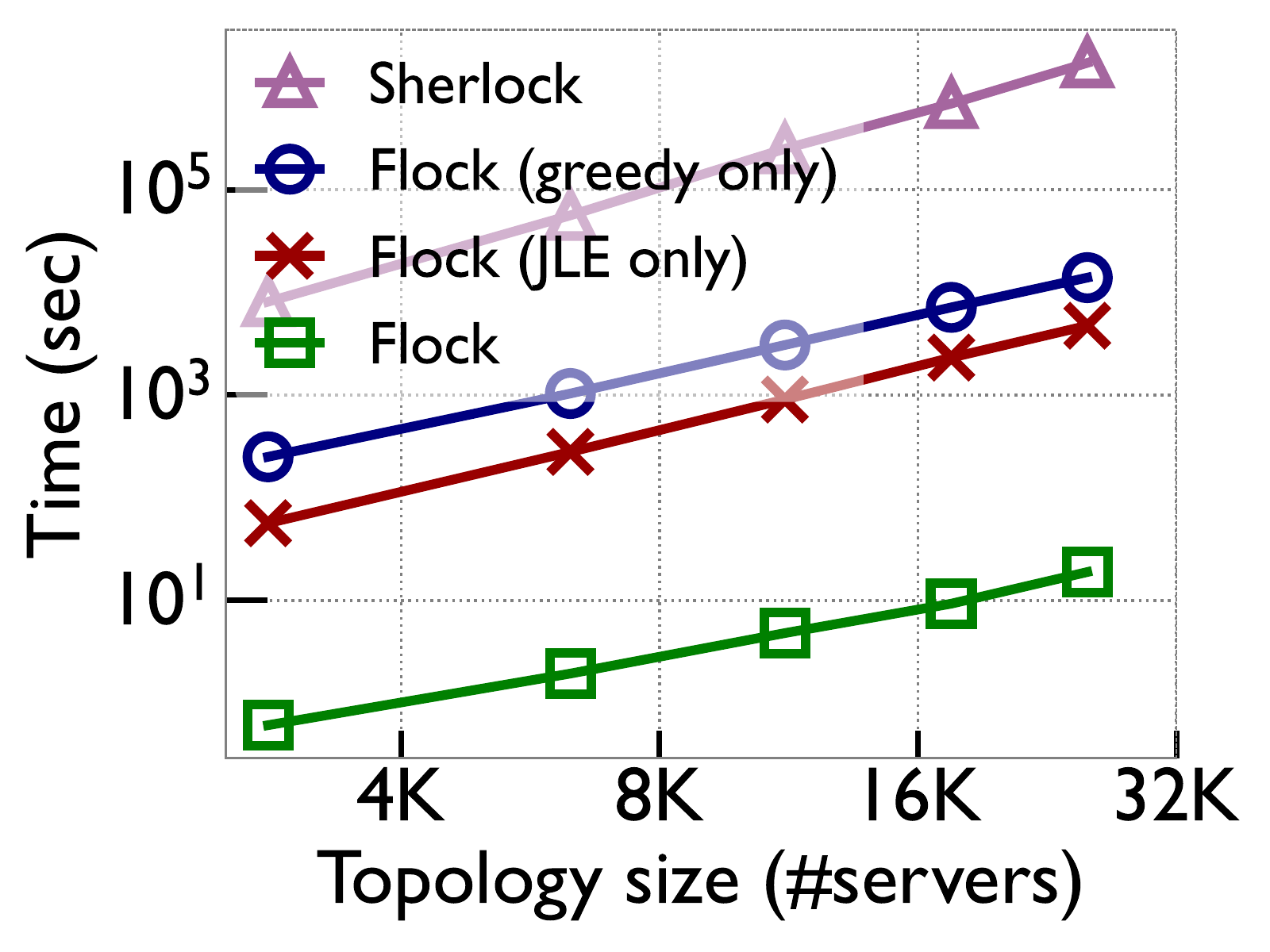}
    \caption{Benefit of greedy and JLE}
    \label{fig:jleBenefit}
  \end{subfigure}
  \end{minipage}
  \hskip 0.15cm
     \begin{minipage}[b]{0.25\textwidth}
    \centering
    \begin{subfigure}[b]{\linewidth}
    \includegraphics[width=1.24\textwidth]{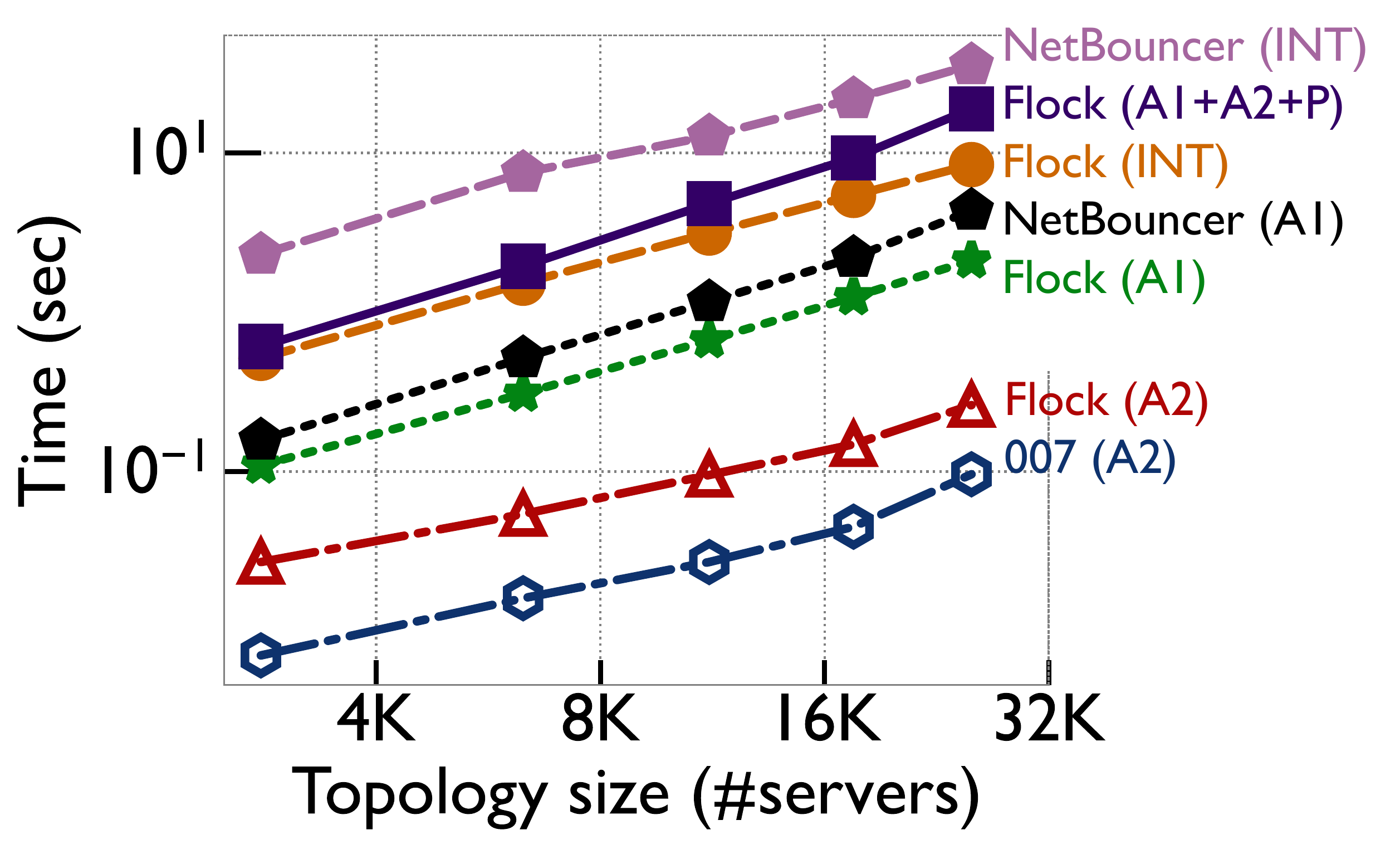}
    \caption{Scheme runtime}
    \label{fig:scaleRuntime}
  \end{subfigure}
  \end{minipage}
  \caption{(a), (b) Accuracy on failure scenarios in testbed (solid markers). For (a), for comparison, we also show precision recall tradeoff curves with recalibrated parameters (hollow markers) 
    (c) Running time vs. a past PGM scheme (Sherlock); Flock achieves the same accuracy while being >$10^4$x faster. Also shown is the effect of Flock's two optimizations (JLE, greedy) alone.
    (d) Running time of all schemes on various topology sizes.
  }
\label{fig:testbedExperiments}
\end{figure*}



\begin{figure*}
    \hskip -0.6 cm
    \begin{minipage}[b]{0.24\textwidth}
        \centering
        \begin{subfigure}[b]{\linewidth}
            \includegraphics[width=1.10\textwidth]{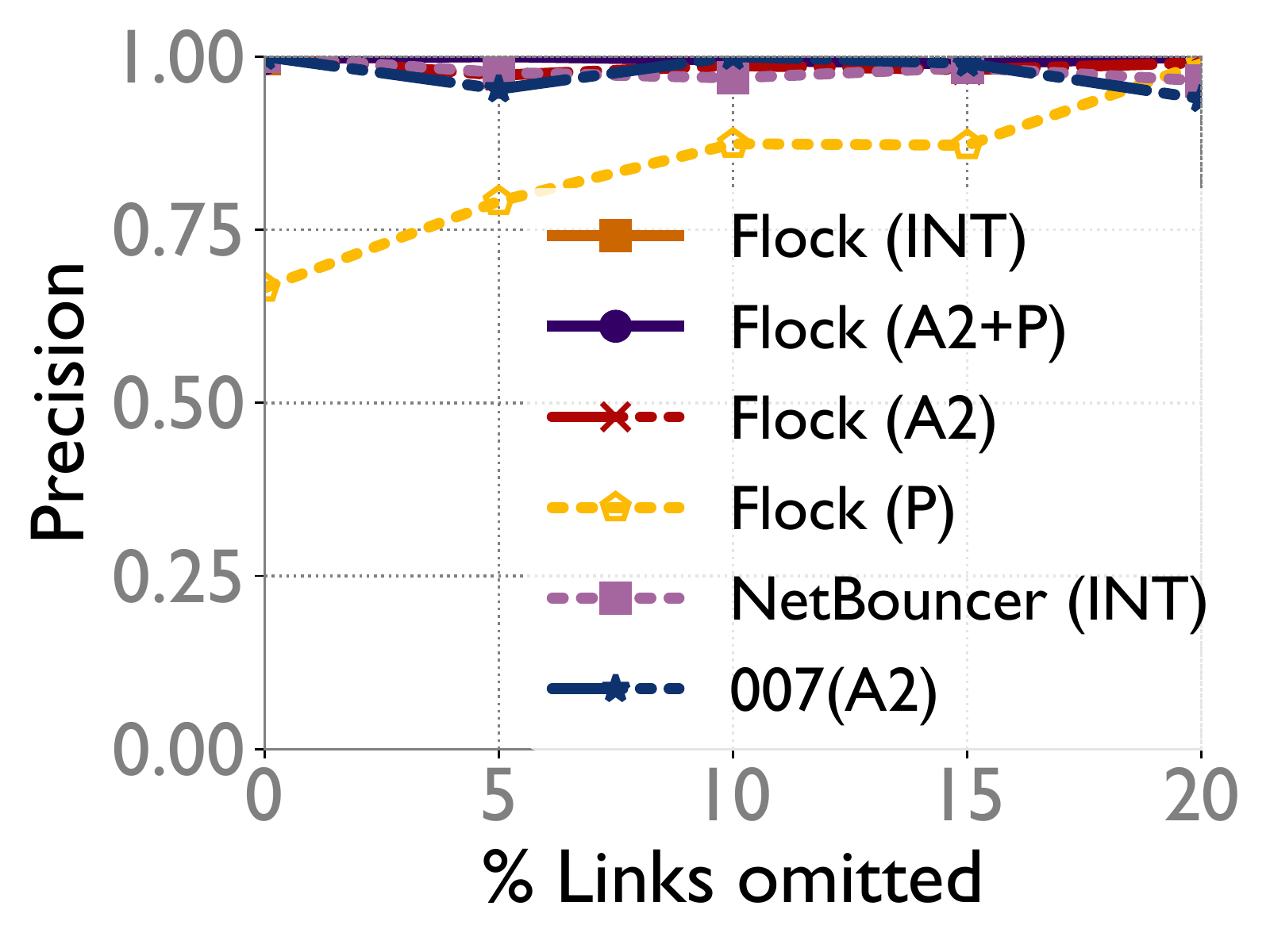}
            \caption{Precision}
            \label{fig:omitPrecision}
        \end{subfigure}
    \end{minipage}
    \hskip 0.2 cm
    \begin{minipage}[b]{0.24\textwidth}
        \centering
        \begin{subfigure}[b]{\linewidth}         \label{fig:irregularrity}
            \includegraphics[width=1.10\textwidth]{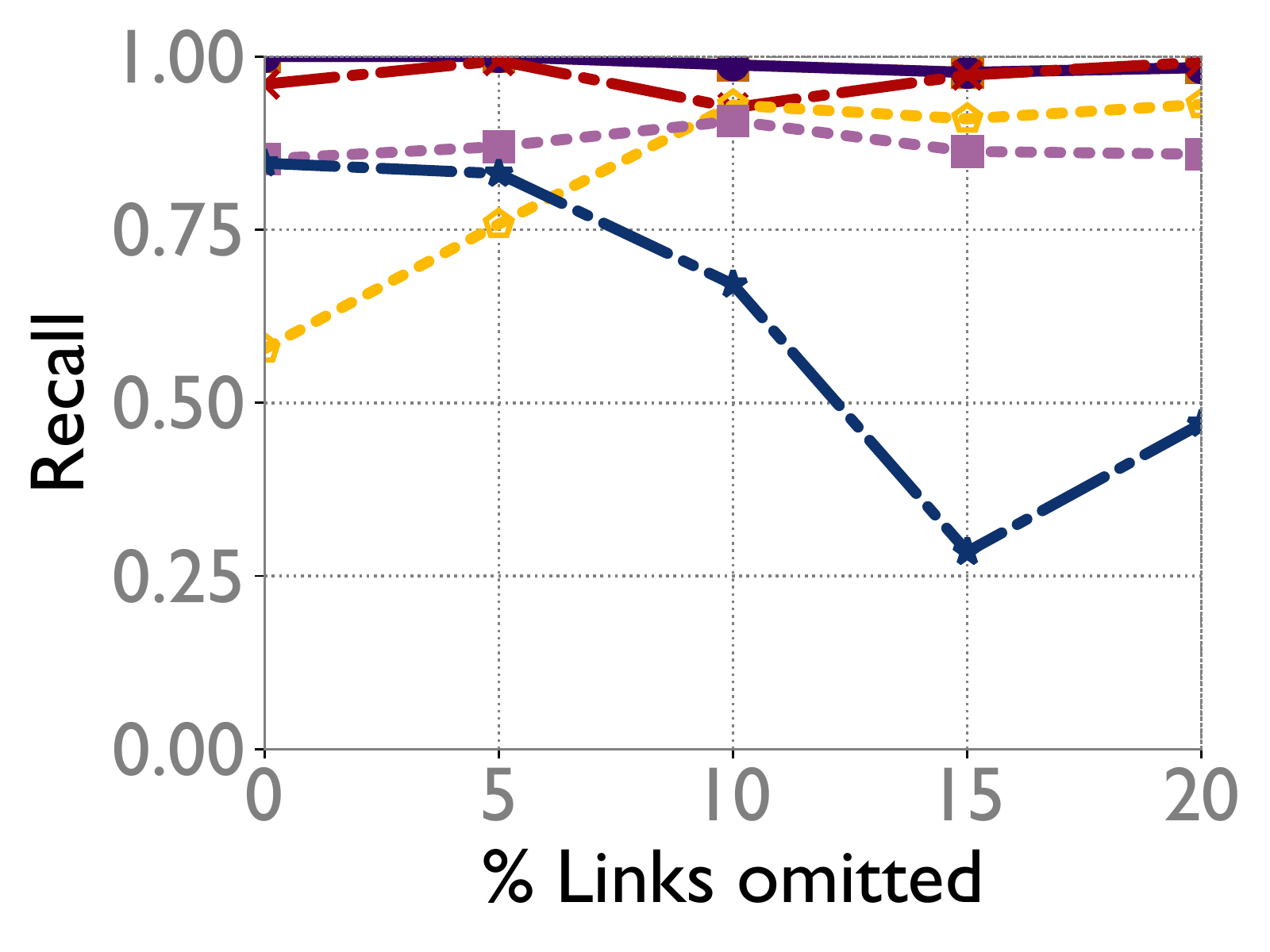}
            \caption{Recall}
            \label{fig:omitRecall}
        \end{subfigure}
    \end{minipage}
    \hskip 0.1 cm
    \begin{minipage}[b]{0.24\textwidth}
        \centering
        \begin{subfigure}[b]{\linewidth}
            \includegraphics[width=1.10\textwidth]{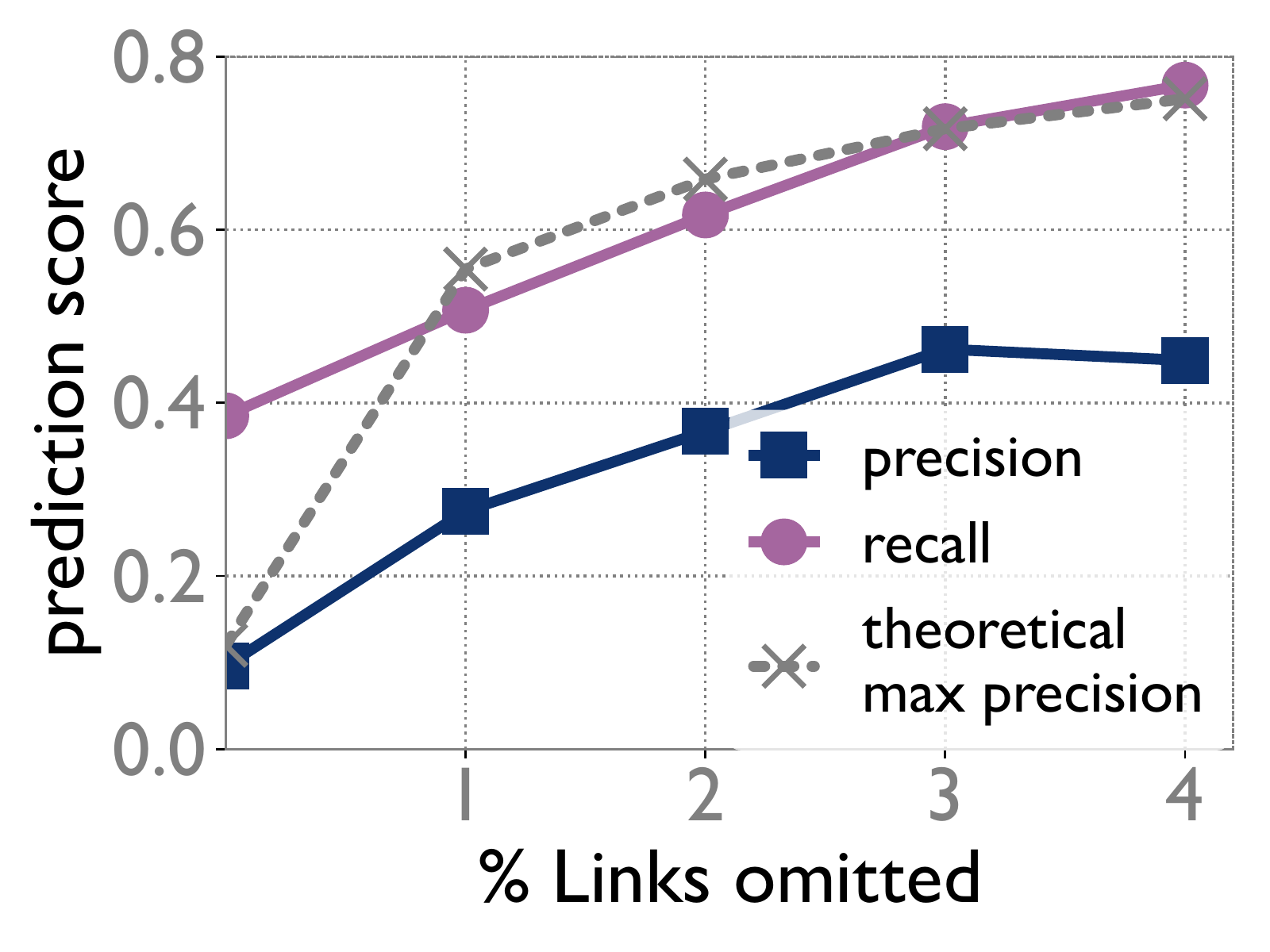}
            \caption{Flock(P) on a hard scenario}
            \label{fig:flock_p_hard_scenario}
        \end{subfigure}
    \end{minipage}
    \caption{(a),(b) Accuracy on ``irregular'' Clos networks with a few links omitted. 
    (c) Flock (P) produces useful results in a difficult scenario, where other schemes don't apply.
    }
    \label{fig:flockIrregularRuntime}
\end{figure*}



\subsection{Link flaps}
\noindent We use a per-flow analysis (\S~\ref{inference graph model}), classifying a flow as problematic if its RTT is > 10 ms. Since the analysis is per-flow and not per-packet, we had to recalibrate parameters. Flock does not model acks traversing the reverse path, which is important in this case. Accounting for that with a revised model would be possible, we leave that for future work.
Even with a somewhat inaccurate model, Flock (INT) reduces the error rate by 1.66x over NetBouncer (INT) (fscore: 0.81 vs 0.69) and Flock (A2) reduces the error rate over 007 by 1.8x (see Fig.~\ref{fig:testbed_link_flap}).






\subsection{Irregular Clos topologies}
\noindent Real world datacenters are rarely perfectly symmetric like a Clos topology and typically have asymmetries due to failures, policies, piecemeal upgrades, etc. To see the effect of topology irregularity, we 
omit links from the fat tree. 
We recalibrated parameters on the irregular topologies for all schemes since the topology can be known in advance. 
(we also tested without the recalibration; \Sys's improvement over other schemes was even higher in this case. We omit the results for brevity.)

Fig.~\ref{fig:omitPrecision},~\ref{fig:omitRecall} show that Flock's accuracy is robust to topology irregularity. 007 is sensitive to topology irregularity, probably because its effect is similar as having traffic skew. We omit A1 since its active probing mechanism is designed only for regular Clos topologies.

\parheading{Flock with passive only input (P):} Some networks may only have passive information available. Past fault localization schemes can not be applied to passive only input since they don't handle path uncertainty.  Operators in this situation resort to manual troubleshooting (traceroutes, adjustments to routing or taking links offline, etc.) which can take days. Flock with only passive input (P) can provide partial analysis (Fig.~\ref{fig:omitPrecision},~\ref{fig:omitRecall}). 
Interestingly, Flock (P)'s accuracy actually \emph{improves} with more links removed.
This is because in a symmetric Clos topology, there are equivalence classes of links (e.g. all links from a leaf switch to the spine layer) that cannot be differentiated because they participate in the same ECMP paths. As the topology becomes irregular, this breaks symmetry, and Flock's inference algorithm automatically takes advantage of this.

Fig.~\ref{fig:flock_p_hard_scenario} shows an even more difficult fully passive scenario where the failed link is one of several symmetric links in a Clos topology, the network has little irregularity for Flock to leverage (< 5\% omitted links) and absence of active probes or path tracing. Flock (P) achieved >75\% recall and >40\% precision. We also show the theoretical maximum precision (calculated from the topology's link equivalence classes). Note that 40\% precision means Flock has narrowed down the faulty link to about 2-3 possibilities. We believe this will be an extremely helpful starting point for operators.


\begin{center}
\begin{table}
\scriptsize
\centering
\setlength\tabcolsep{2.0pt}
\begin{tabular}{|c|c|c|c|c|c|c|c|c|c|c|}
  \hline
  \multicolumn{2}{|c|}{\begin{tabular}{@{}c@{}}Parameters \\ calibrated for $\rightarrow$\\ (D: different, S: same) \end{tabular}} &
  \multicolumn{2}{c|}{\begin{tabular}{@{}c@{}}Different\\ topology \end{tabular} } & 
  \multicolumn{2}{c|}{\begin{tabular}{@{}c@{}}Different\\ failure \\ rate \end{tabular} } & 
   \multicolumn{2}{c|}{\begin{tabular}{@{}c@{}}Different\\ monitoring \\ interval \end{tabular} } & 
\multicolumn{2}{c|}{\begin{tabular}{@{}c@{}}Different\\ failure \\ scenario \end{tabular} } &
  \begin{tabular}{@{}c@{}}\textbf{Aggregate}\\ \textbf{score} \\ (average Fscore)\end{tabular} \\ \clineB{1-10}{1.0}
  
  \multicolumn{2}{|c|}{\begin{tabular}{@{}c@{}} p: precision, r: recall \end{tabular}} & p & r & p & r & p & r & p & r & \\ \thickhline

  \multirow{2}{*}{\begin{tabular}{@{}c@{}}Flock\\ (A1+A2+P) \end{tabular}} & D & 0.92 & 0.98 &  0.97 & 1 &	0.98 & 1 & 0.95 & 0.99 & 0.973\\ \clineB{2-11}{1.0}
   & S  & 0.96 & 0.98 & 0.97 & 1 & 0.99 & 1 & 0.96 & 0.99 & 0.981\\ \thickhline

  \multirow{2}{*}{Flock (A2)}    & D  & 0.90 & 0.99 &  0.96 & 1 &	0.86 & 0.97 &	0.94 & 0.98 & 0.948\\ \clineB{2-11}{1.0}
  & S & 0.94 & 0.98 &  1 & 1 &	0.94 & 0.92 &	0.94 & 0.98 & 0.962\\ \thickhline

  \multirow{2}{*}{Flock (INT)}  & D   & 0.92 & 0.99 & 0.96 & 1 &	0.98 & 1 &	0.95 & 0.99 & 0.973 \\ \clineB{2-11}{1.0}
  & S & 0.96 & 0.99 & 0.96 & 1 &	0.99 & 1 &	0.96 & 0.99 & 0.981\\ \thickhline

  \multirow{2}{*}{007 (A2)}   & D     & 0.75 & 1  &	0.87 & 1 & 0.51 & 0.82 &	1 & 0.33 & 0.728 \\ \clineB{2-11}{1.0}
  & S & 0.82 & 0.74 & 1 & 1 & 0.47 & 0.87 & 0.82 & 0.74 & 0.792\\ \thickhline
                  
  \multirow{2}{*}{\begin{tabular}{@{}c@{}}NetBouncer\\ (INT) \end{tabular}} & D & 0.25 & 0.33	& 0.95 & 1 &	1 & 0.66 & 0.28 & 0.33 & 0.589\\ \clineB{2-11}{1.0}
  & S & 0.81 & 0.9 &  1 & 1 & 0.95 & 0.85 & 0.81 & 0.9 & 0.901 \\ \hline
 \end{tabular}
 \caption{Evaluating parameter robustness. For each scheme, we show accuracy when its parameters are calibrated on a different environment than the test data set (D) and when they are calibrated in the same environment (S).}
 \label{paramRobust}
\end{table}
\end{center}

\subsection{Parameter calibration robustness} \label{parameterSensitivity}
\noindent As discussed in \S~\ref{sec:paramCalibration}, all the schemes we consider have parameters that must be set, and we have calibrated them based on simulations with known ground truth. What happens when these systems encounter unexpected situations?

To test robustness, we trained each scheme on one environment and tested in a different environment. (This is effectively a strong form of cross-validation where not only are the train and test sets different, they are drawn from \emph{different distributions}.) Specifically, we created different types of differences between train and test, changing the (a) duration of monitoring, (b) topology, (c) failure rate (training set has  failed links with significantly different drop rates), 
and (d) failure type. In particular, for (b), the schemes were calibrated in our simulator with random packet drops, and tested on misconfigured queues in a 20x smaller topology in our physical testbed.
Table~\ref{paramRobust} shows the accuracy in these cases, both when the train and test sets are drawn from different distributions (``D'') and when they are drawn from the same distribution (``S''). The table's aggregate score column shows Flock was fairly robust to parameter calibration in a different environment, with under 2\% loss in accuracy. 007 was also robust ($6\%$ loss) while NetBouncer was more sensitive ($31\%$ loss).

We also tested Flock's parameter sensitivity, i.e., how precision and recall vary with perturbations of its parameters. Accuracy remained high for many choices of parameters. See \ifEurosys Fig.7a in appendix of ~\cite{FlockFullDraft} \else Fig. ~\ref{fig:paramSensitivityPgPb} in appendix \fi.



\subsection{Running time and scalability} \label{runningTimeSection}

\parheading{Algorithm runtime:} \Sys's main algorithmic innovation is its fast PGM inference compared to Sherlock's PGM inference. Fig.~\ref{fig:jleBenefit} shows \Sys's inference is more than 4 orders of magnitude of faster than Sherlock, whose runtime on a large network was estimated to be 19 days, based on extrapolating a partial run. Recall \Sys employs two optimizations: greedy and JLE. Fig.~\ref{fig:jleBenefit} shows each of these optimizations alone yields a $\approx 100$x improvement over Sherlock. 

Fig.~\ref{fig:scaleRuntime} compares \Sys to the non-PGM schemes. Flock is faster than NetBouncer on the same input data. 007 is the fastest but its time savings (< 1 sec) does not trade-off well with accuracy.

\parheading{Agent/Collector} \label{collectorScalability}: Refer to Appendix  \ref{collectorScalabilityAppendix} for the scalability of our agent/collector. As other commercial solutions also exist~\cite{manageengine,solarwinds}, we leave more optimized agent/collector designs for future work.

\section{Related work}
Many other works have studied fault localization outside of datacenter networks -- for troubleshooting reachability, black holes in IP networks~\cite{maxCoverage,tomo,netscope,netsonar}, virtual disk failures~\cite{deepView}, performance problems in distributed services~\cite{sherlock,gestalt,sage} and application performance anomalies~\cite{snap}. Some of these can benefit from PGM-based inference, accelerated via JLE. We leave this for future work.

Detecting entire flow drops, for e.g., caused by a misconfigured ACL or a forwarding loop, is challenging for end-to-end schemes (as noted in~\cite{netbouncer}), partly due to (un)available input when paths are fully blocked. Other schemes such as NetSeer~\cite{netseer} and Omnimon~\cite{omnimon} or network verification~\cite{mai2011debugging,kazemian2012header,khurshid2013veriflow,fogel2015general} are more suitable for this class of faults.

Several works orchestrate active probes for inference~\cite{score,shrink,nodeFailureLocalizationViaNetworkTomography,lend,algebraicApproach,KDD-Fault-Localization,netbouncer,netnorad,pingmesh}. Flock can handle active probes and outperforms one such approach (NetBouncer). Additionally, it can use passive data for accuracy gains. deTector~\cite{deTector}, MaxCoverage~\cite{maxCoverage}, Tomo~\cite{tomo} and Score~\cite{score}  find a minimal set of components that explain most of the problems (e.g. packet drops). 
We expect them to run into similar problems as 007, since they don't account for traffic skew. Simon~\cite{simon} reconstructs queuing times from active probes. This allows it to diagnose high latency, but not silent packet drops.  Packet mirroring~\cite{everflow} can catch packet drops that happen in the switch pipeline (e.g. congestion drops), but can not  detect silent interswitch or silent intercard drops. Flock is well suited to detect such problems. Pingmesh~\cite{pingmesh} and NetNorad~\cite{netnorad} use active pings, but do not provide complete localization. \cite{FbLocalization} identifies anomalies among symmetric links in a Clos network using statistical tests. It is sensitive to topology irregularity and requires path information. 

\section{Conclusion}
Flock is a fault localization system for large datacenter networks based on end-to-end information.  \Sys's key innovation is an optimized MLE inference algorithm which allows it to use a PGM at scale, achieving both high accuracy and speed, where past work achieved only one of the two.

\section*{Acknowledgements}
We thank Radhika Mittal and members of the SysNet group at UIUC for providing feedback on an earlier version of the paper. We also thank CoNext 2023 reviewers for their reviews and feedback. 

\bibliographystyle{abbrv}
\balance
\begin{small}
\bibliography{bibliography}
\end{small}

\ifEurosys
\else

\appendix

\begin{figure*}
    \hskip -0.6 cm
    \begin{minipage}[b]{0.4\textwidth}
        \centering
        \begin{subfigure}[b]{\linewidth}
            \includegraphics[width=0.8\textwidth]{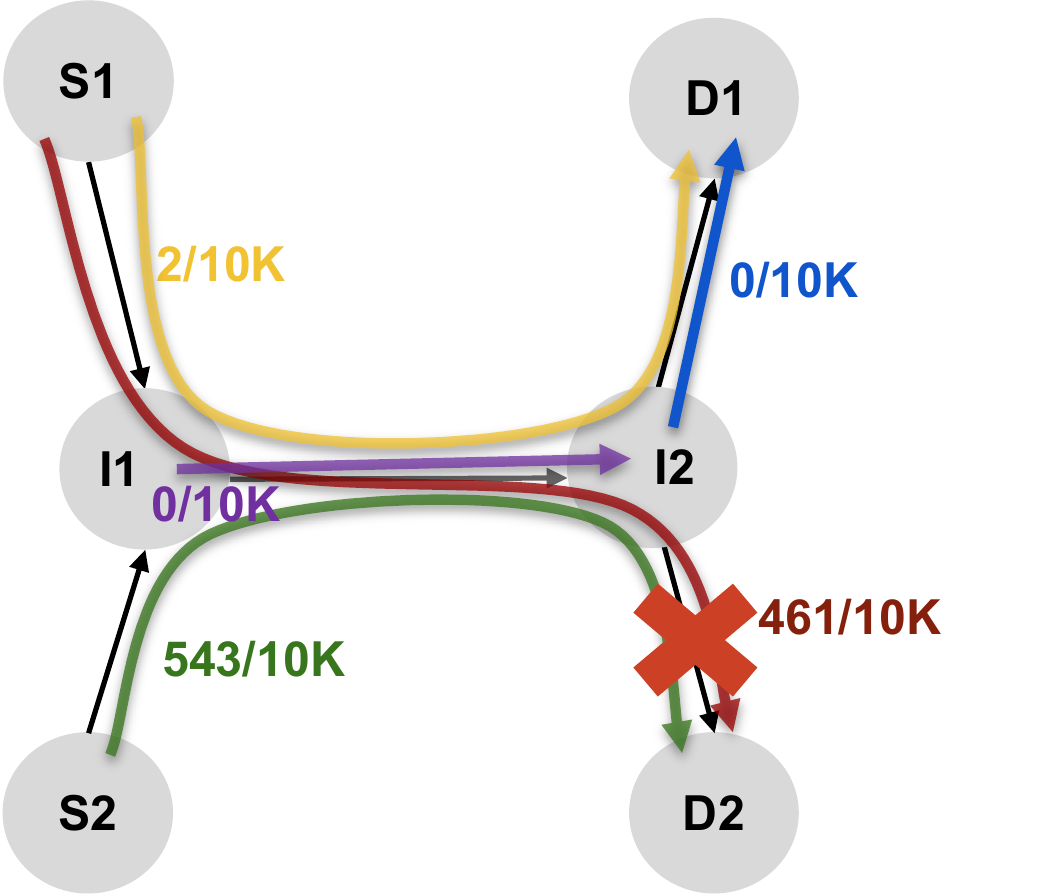}
            \caption{Example: failed link shown via the red cross}
            \label{fig:exampleNetwork}
        \end{subfigure}
    \end{minipage}
    \hskip 0.7 cm
    \begin{minipage}[b]{0.4\textwidth}
        \centering
        \begin{subfigure}[b]{\linewidth}
            \includegraphics[width=1.15\textwidth]{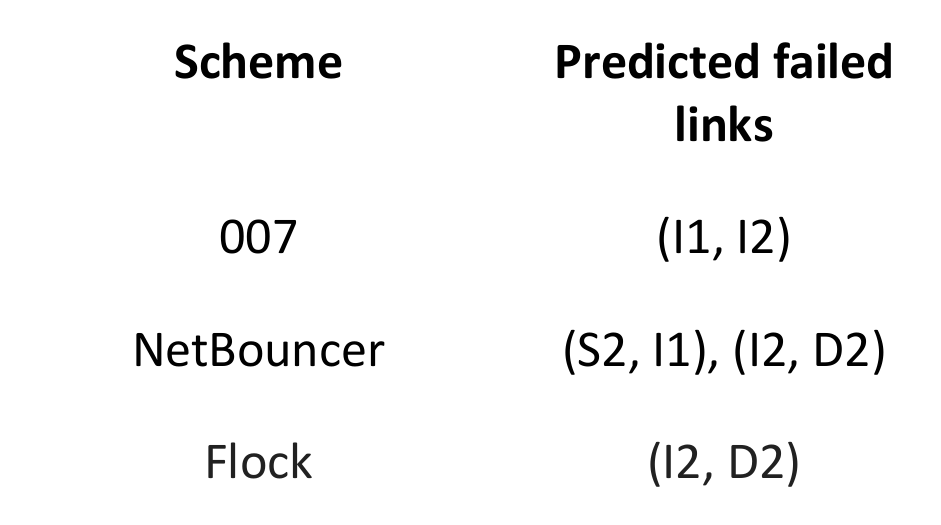}
            \caption{Output of various schemes}
            \label{fig:exampleLegend}
        \end{subfigure}
    \end{minipage}
  \caption{Example: (a) 5 links in the network. 5 flows shown in 5 different colors, annotated with packets dropped/packets sent. (b) Flock correctly localizes the failed link.}
  \label{fig:example}
\end{figure*}

\section{Proofs} \label{greedyJustificationProofs}

If the topology, flow statistics and path taken for each flow are chosen arbitrarily, then finding the MLE for such adversarial inputs, unsurprisingly, is NP-hard (\S~appendix~\ref{greedyJustificationProofs}). 

\begin{definition}
Define adversarial inference as inference when the input is  arbitrary, consisting of topology, flow metrics with arbitrarily chosen source and destinations, paths and arbitrarily chosen flow metrics (packets sent/dropped). 
\end{definition}

However, we can make the following assumption, in the context of packet drops, to alleviate intractability for adversarial inputs: for each link, there is a (unknown) ground truth drop probability. Rather than adversarially, each packet crossing that link is dropped independently according to the drop probability of the link. First, we prove that adversarial inference is NP-hard.

\begin{theorem}\label{adversarialVersionNPHardAppendix}
Adversarial inference is NP-hard.
\end{theorem}

\begin{proof}
We reduce the problem of finding a minimum vertex cover in a graph to the adversarial inference problem. For the proof, we will design an algorithm for min-vertex cover that makes polynomial number of accesses to an oracle $O^{AI}$ for adversarial inference.

Let's say, we're given a graph G = (V, E) where V is the set of vertices and E is the set of edges and our goal is to find a minimum vertex cover in this graph. We create topology $\mathcal{T}$ for $O^{AI}$ as follows-
\begin{itemize}[leftmargin=*]
\item We create ``vertex''-nodes $n_v^1$ and $n_v^2$ in $\mathcal{T}$ for each vertex $v \in V$ and attach them with a directed link $l_v$: $(n_v^1 \rightarrow n_v^2)$. Note that all vertex-nodes in $\mathcal{T}$ have exactly one link, either incoming or outgoing.

\item For each edge $e \in E$, we create an ``edge"-flow $f_e$, where $f_e$ goes through links $(l_{v_1}$ and $l_{v_2})$, where $v_1$ and $v_2$ are the endpoints of $e$. In order for this to be a legitimate path, we connect the endpoints of $(l_{v_1}$ and $l_{v_2})$ to a special node. 
\item For each link $l$ in $\mathcal{T}$, we create a ``link"-flow $f_l$. going through just link $l$.
\end{itemize}

Finally we need to assign each flow, some number of packets sent and dropped. We first derive an expression for the return value of oracle $O^{AI}$. The output of $O^{AI}$ is a binary assignment to links in $\mathcal{T}$ where we interpret a value of 0 as that link being failed and 1 as the link being up. Let $E_{\mathcal{T}}$ be the number of links in $\mathcal{T}$ and $\gamma^A_f =  \prod_{l\in Path(f)} l$ denote the status of the path taken by flow a $f$ with $\gamma_f^A$=0 being path failed and 1 being path not failed. Recall that a path is deemed to be failed if it contains at least one failed link. For flow $f$, if $n$ and $r$ be the number of packets sent/dropped by the flow respectively, then the likelihood for flow $f$ can be written as:



\begin{align*}
& \frac{P[f|H = \{l_1, l_2, ...\ l_{E_{\mathcal{T}}}\}]}{P[f|l_1=1, l_2=1, ...\ l_{E_{\mathcal{T}}}=1]} \quad \quad  \\
&\quad \quad \quad = \frac{\gamma^A_f p_g^r (1 - p_g)^{n-r} + (1 - \gamma^A_f) p_b^r (1 - p_b)^{n-r}}{p_g^r (1-p_g)^{n-r}}\\
&\quad \quad \quad = \frac{p_b^r (1-p_b)^{n-r}}{p_g^r (1-p_g)^{n-r}}\bigg(1+\gamma^A_f\Big(\frac{p_g^r (1-p_g)^{n-r}}{p_b^r (1-p_b)^{n-r}} - 1\Big)\bigg)\\
&\quad \quad \quad  = \frac{1+\alpha_f \gamma^A_f}{\alpha_f + 1} = \frac{(1 + \alpha_f \prod_{l\in Path(f)} l)}{\alpha_f + 1} 
\end{align*}

where, $\alpha_f = \frac{p_g^r (1-p_g)^{n-r}}{p_b^r (1-p_b)^{n-r}} - 1 \in (-1, \infty)$ is a constant for flow $f$ irrespective of $H$. The oracle $O^{AI}$ then returns 

\[\underset{H\in  \{0,1\}^{E_{\mathcal{T}}}}{\text{arg max}} \prod_{f:\text{flows}}P[f| H] = \underset{H\in  \{0,1\}^{E_{\mathcal{T}}}}{\text{arg max}} \prod_{f:\text{flows}} (1 + \alpha_f \prod_{l\in Path(f)}l^H) \]

Where $l^H$ is the status (0/1) of link $l$ as per hypothesis $H$. Given the above expression, we set the number of packets sent/dropped for all flows in the following way
\begin{enumerate}
    \item \label{a1}
    For all edge-flows $f_e$,  we set $n$, $r$: the number of packets sent/dropped such that $1 + \alpha_{f_l} = \frac{1}{C}$ where $C>>1$.
    \item \label{a2}
    For all link-flows $f_l$ where $l$ is connected to a special node, we set $n$, $r$: the number of packets sent/dropped in such a way that if $l$ is connected to a special node, then $1 + \alpha_{f_l} = 1+C$. This ensures that $O^{AI}$ will always assign a label of 1 to a link connected to a special node (recall that 1 denotes the link being up). 
    \item \label{a3} For all link-flows $f_l$ where both endpoints of $l$ are connected to vertex nodes, we set $n$, $r$: the number of packets sent/dropped such that $1 + \alpha_{f_l} = 1+\epsilon$ where $\epsilon$ is a small number > 0. This assigns a small cost of assigning a label 0 to a link in $\mathcal{T}$ whose both endpoints of $l$ are connected to vertex nodes.
\end{enumerate}

 The vertex cover is obtained by simply picking vertices in $G$ corresponding to links in $\mathcal{T}$ that are deemed as failed by $O^{AI}$. Conditions (\ref{a1}) and (\ref{a2}) above ensure that only vertex-links will be classified as failed by $O^{AI}$. Conditions (\ref{a2}) and (\ref{a3}) ensure that the result set of vertices will cover all edges. Condition (\ref{a3}) ensures that the resultant vertex cover will be of the smallest size. \end{proof}

\begin{figure*}
  \hskip -0.2 cm
    \begin{minipage}[b]{0.31\textwidth}
    \centering
  \begin{subfigure}[b]{\linewidth}
    \includegraphics[width=0.95\textwidth]{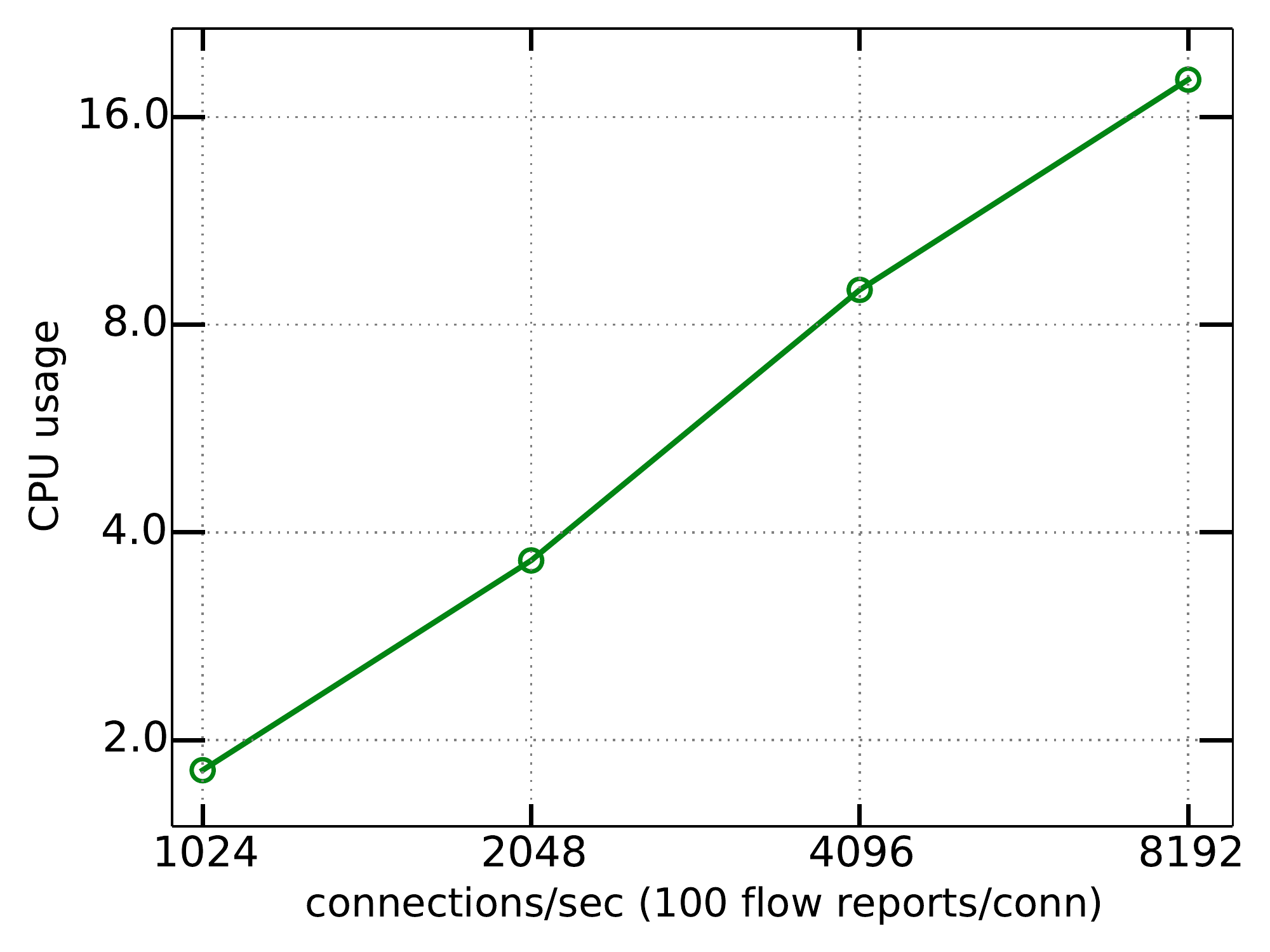}
    \caption{Collector scaling}
    \label{fig:collectorScale}
  \end{subfigure}
  \end{minipage}
  \begin{minipage}[b]{0.31\textwidth}
    \centering
    \hskip -.2 cm
  \begin{subfigure}[b]{\linewidth}
    \includegraphics[width=1.05\textwidth]{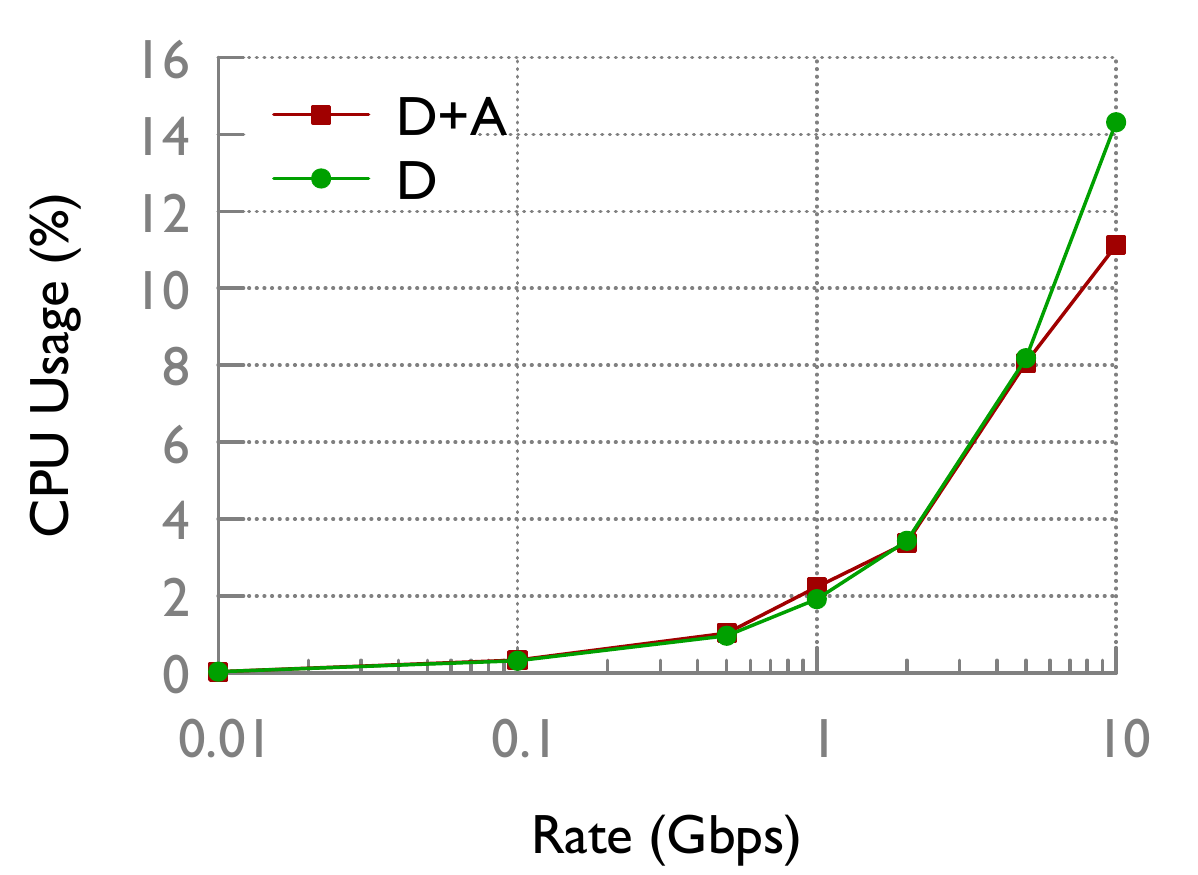}
    \caption{agent CPU usage: single flow}
    \label{fig:singleFlowAgentCost}
  \end{subfigure}
  \end{minipage}
\begin{minipage}[b]{0.31\textwidth}
    \centering
    \begin{subfigure}[b]{\linewidth}
    \includegraphics[width=1.05\textwidth]{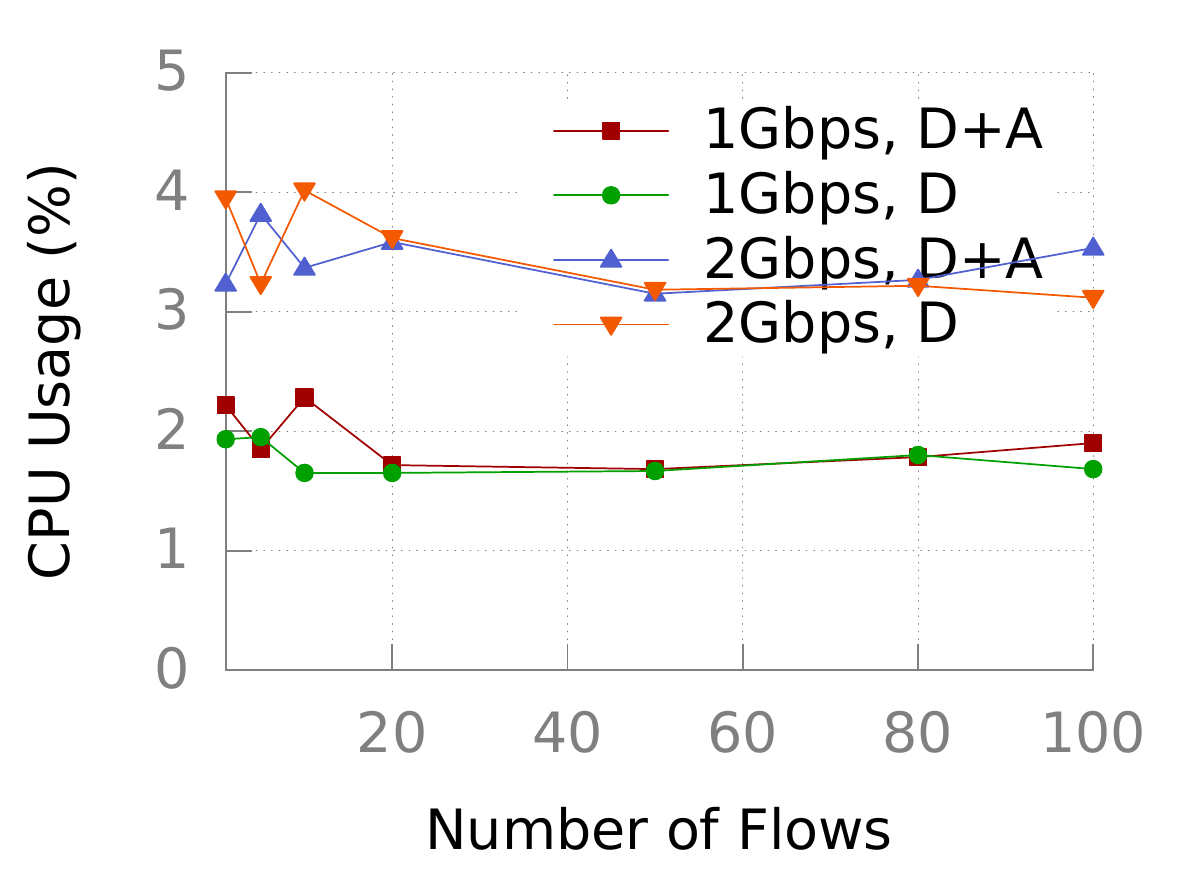}
    \caption{several flows}
    \label{fig:multipleFlowAgentCost}
  \end{subfigure}
  \end{minipage}
  \caption{(a) CPU usage at the collector, varying number of agents (b) CPU usage on the agent for handling a single flow. D: cost of packet header dumping, A: cost of compiling flow reports from packet headers.(c) agent CPU usage with many concurrent flows }
\label{fig:agentCollectorScaling}
\end{figure*}

\parheading{Agent/Collector scalability} \label{collectorScalabilityAppendix}: Our agent's CPU usage  
for a end-host sending traffic 
grew linearly with the data rate (see Fig~\ref{fig:singleFlowAgentCost}) and was
<2\% of one core for a 1Gbps uplink and 10-15\% of a core for a 10 Gbps uplink. 
As can be seen from figure~\ref{fig:multipleFlowAgentCost}, the resource usage was independent of the number of flows.
We verified that our multicore collector can handle 8K connections/sec from agents (see Fig~\ref{fig:agentCollectorScaling}). We tested this by launching several agent processes that generate dummy flow reports to send to the collector . 
These results show that the passive  monitoring can be handled by an end-host agent and a centralized collector.

\begin{figure}
\centering
   \begin{minipage}[b]{0.27\textwidth}
        \centering
        \begin{subfigure}[b]{\linewidth}
            \includegraphics[width=1.0\textwidth]{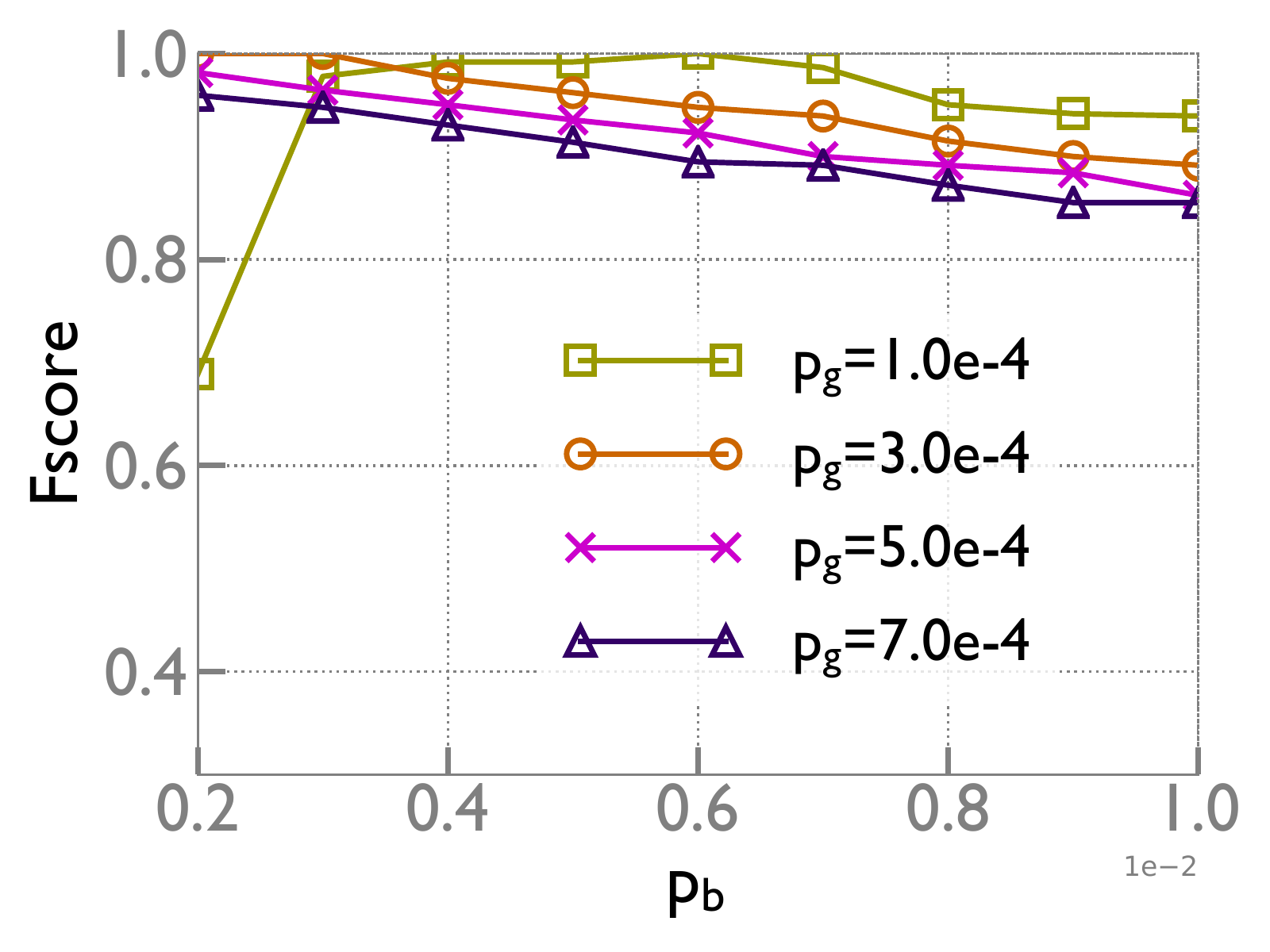}
            \caption{Parameter sensitivity}
            \label{fig:paramSensitivityPgPb}
       \end{subfigure}
    \end{minipage}
     \begin{minipage}[b]{0.27\textwidth}
    \centering
    \begin{subfigure}[b]{\linewidth}
    \includegraphics[width=1.0\textwidth]{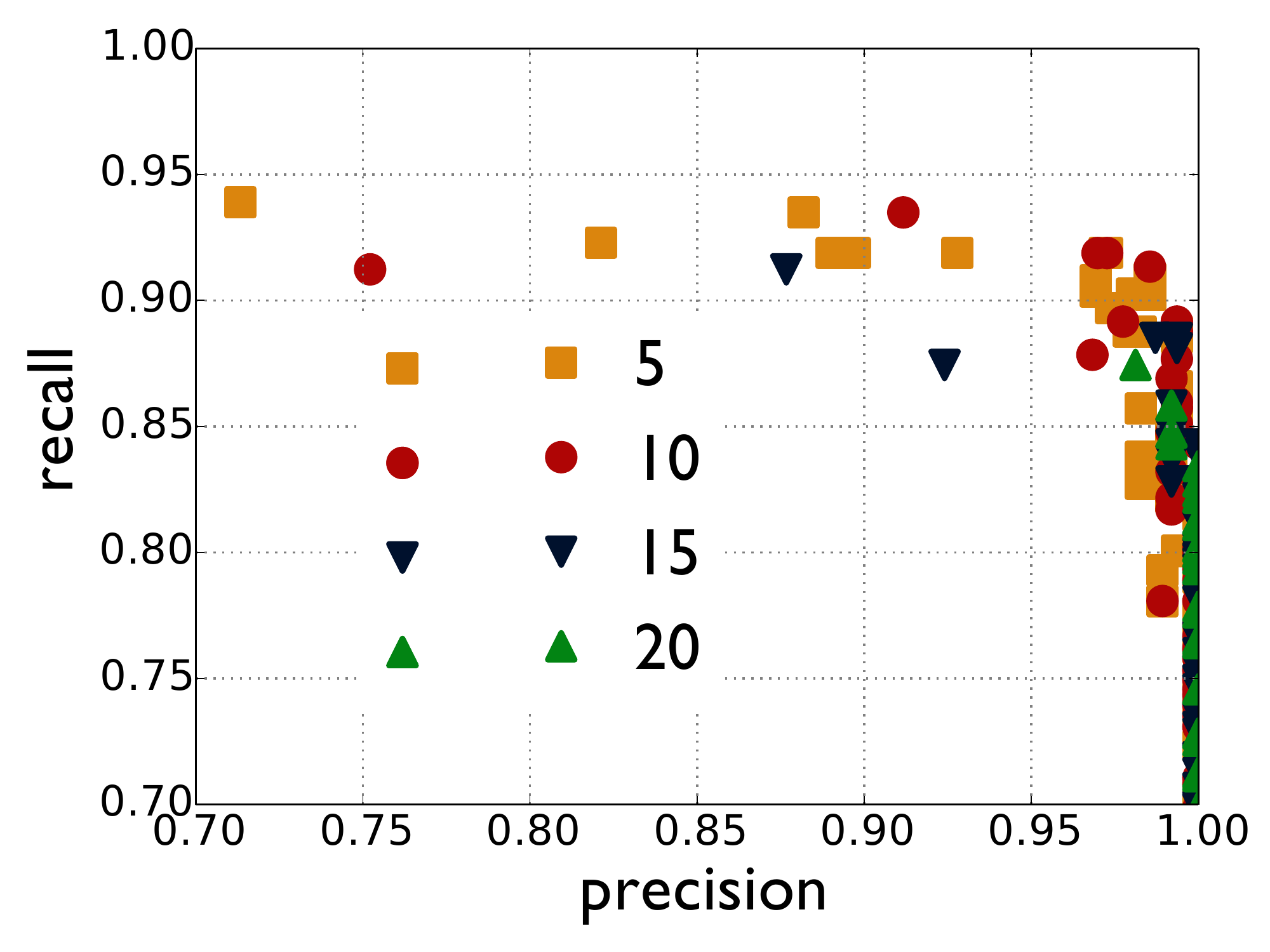}
    \caption{Effect of changing priors}
    \label{fig:fig:paramSensitivityPrior}
  \end{subfigure}
  \end{minipage}
\caption{(a) Effect of changing $p_g$ and $p_b$: intuitively, we expect precision to increase when $p_g$ or $p_b$ is increased, at the cost of reduced recall, which is confirmed by the figure (b) Higher priors result in points to the right. Priors $\rho$ resulted in a significant reduction in false positives. 
  }
  \label{fig:paramSensitivity} 
\end{figure}

\noindent \begin{theorem} \label{greedyGuaranteeConditionsAppendix} 
For any topology with $(1/\alpha)$-skewed traffic, \Sys's inference returns the set of all failed links if the number of failures is $\leq \alpha/2$ with high probability, the number of packets $T_{min}$ crossing every link is larger than a certain threshold, and the drop probabilities are $<p_g$ on all good links and $>p_b$ on all failed links where ($p_g$, $p_b$) satisfy the condition $5p_g < p_b < 0.05$.
\end{theorem}

The condition $T_{min} > T_0$ ensures that the effect of variance in the number of packets dropped by a link is small. Theorem~\ref{greedyGuaranteeConditions} holds for a wide range of values of $p_g$ and $p_b$ (see lemma~\ref{paramConditions}).

\begin{proof}

If $\gamma^H_f$ be the 0/1 status of the path taken by flow $f$ as per hypothesis $H$ (0 being that the path has at least one failed link as per $H$ and hence is labeled as failed), we can express the (normalized) log likelihood for a flow given a hypothesis as follows

\begin{align*}
& \log P[f|H = \{l_1, l_2, ...\ l_n\}] - \log P[f|l_1=1, l_2=1, ...\ l_n=1]  \\
=& \log \big(\gamma^H_f p_g^r (1 - p_g)^{n-r} + (1 - \gamma^H_f) p_b^r (1 - p_b)^{n-r}\big) \\
& - \log \big(p_g^r (1-p_g)^{n-r}\big)\\
=& (1 - \gamma^H_f) \Big(r \log \frac{p_b}{p_g} + (n-r)\log \frac{(1-p_g)}{(1-p_b)}\Big) \\
=& (1 - \gamma^H_f) \lambda (r - n\mu ) = (1 - \gamma^H_f) \lambda \sum_{i=0}^{n} (b_i - \mu)  \\
=& (1 - \gamma^H_f) X(f, H)
\end{align*}

where $\lambda= \log \frac{p_b(1-p_g)}{p_g(1-p_b)} $ and $\mu=\log \frac{(1-p_g)}{(1-p_b)}\big/\log \frac{p_b(1-p_g)}{p_g(1-p_b)}$ are constants and $b_i$ is a binary variable which is 1 if the $i^{th}$ packet of the flow was dropped and 0 otherwise. One can check that $p_g < \mu < p_b$ for any $0 < p_g < p_b < 1$ by taking partial derivatives or using a plotting tool. Note that we can simply ignore the constant $\lambda$ for the purpose of log likelihood maximization.  If the maximum allowed drop rate on a correctly working link be $p^*$ and the maximum path length for the given topology be $k$, then we set $p_g \geq kp^*$ so that for any flow $f$ that does not go through any of the failed links , $(p_f - \mu) < 0$, where $p_f$ is the packet drop probability of the path taken by $f$. We show the following lemma which holds for any reasonable settings for $p_g$ and $p_b$ for the purpose of detecting packet drops. The expected value of the $X(f, H)$ is given as

\begin{align*}
E[X(f,H)] &= E\Big[\lambda \sum_{i=0}^{n} (b_i - \mu)\Big] \\
&= \lambda \sum_{i=0}^{n} E[(b_i - \mu)] = n\lambda (p_f - \mu)
\end{align*}

\begin{lemma} \label{paramConditions}
If $\mu = \frac{\log\frac{(1-p_g)}{(1-p_b)}}{\log \frac{p_b(1-p_g)}{p_g(1-p_b)}}$ and $ 5 p_g < p_b \leq 0.05$, then \[0 \leq p_g < \mu < 2\mu < p_b\] 
\end{lemma}
\noindent Lemma~\ref{paramConditions} can be seen by taking partial derivatives or via a numerical plotting tool. 
\begin{lemma} \label{lemmaPathBound}
If $p_l$ denotes the drop probability of link $l$, $L_f =\{l_1, l_2, ... l_k\}$ be the links in the path taken by flow $f$, $p_f$ denotes the drop probability of the path $L_f$ and $H^*$ denote the set of failed links, then $(p_f -\mu) \leq \sum_{(l\in H^*\cap L_f)} p_{l}$ 
\end{lemma}
\begin{proof} For any $l_i \in L_f$, we have- 
\begin{align*}
p_f &= 1 - \Big((1-p_{l_1})...(1-p_{l_{i}})...(1-p_{l_k})\Big)\\
&= 1 - (1-p_{l_{i}})\Big((1-p_{l_1})...(1-p_{l_{i-1}})(1-p_{l_{i+1}})...(1-p_{l_k})\Big)\\
&= 1 - \Big((1-p_{l_1})...(1-p_{l_{i-1}})(1-p_{l_{i+1}})...(1-p_{l_k})\Big) \\
&\quad+ p_{l_{i}}\Big((1-p_{l_1})...(1-p_{l_{i-1}})(1-p_{l_{i+1}})...(1-p_{l_k})\Big)\\
&\leq 1 - \Big((1-p_{l_1})...(1-p_{l_{i-1}})(1-p_{l_{i+1}})...(1-p_{l_k})\Big) + p_l
\end{align*}
Applying the same argument recursively for all links in $H^*\cap L_f$, we get: 
\begin{align*}
 p_f &-\mu = 1 - (1-p_{l_1})(1-p_{l_2})...(1-p_{l_k}) - \mu \\
&\leq \Big(1 - \prod_{l \in  L_f \setminus H^*} (1-p_{l})\Big) - \mu + \sum_{(l\in H^*\cap L_f)} p_{l} \leq \sum_{(l\in H^*\cap L_f)} p_{l} 
\end{align*}

\noindent The last inequality follows from the fact that $(p_f - \mu) < 0$ for a path consisting of only good links. This completes the proof of lemma~\ref{lemmaPathBound}. \end{proof}

Consider the set of hypotheses with single link failures- say $S_1$. We first show that $H_{opt} = \text{arg max}_{H\in S_1} L(H) $, corresponds to a failed link as $T_{min} \rightarrow \infty$  which in turns implies that the greedy algorithm succeeds in picking a failed link in the first iteration. Let $H_l \in S_1$ denote the hypothesis $\{l\}$ where $l$ is a good link in the topology, $F(l)$ is the set of flows that go through the link $l$, $n_f$ is the number of packets sent by flow $f$ and $p_f$ as before is the ground truth drop probability of the path taken by $f$. We have,

\begin{align*}
E[LL(H_l)] = & \sum_{f\in F(l)}E[X(f,H)] = \sum_{f\in F(l)}n(p_f - \mu) \\
    \leq& \sum_{l^* \in H^*} \sum_{f\in F(l)\cap F(l^*)} n(p_f - \mu)\\
    \leq & \sum_{l^* \in H^*}\ \sum_{f\in F(l)\cap F(l^*)} n_f p_{l^*} \leq  \sum_{l^* \in H^*}\ p_{l^*}  \ T(l, l^*) \\
    \leq & \sum_{l^* \in H^*}\frac{1}{\alpha} \ \ p_{l^*} \ T(l^*)\leq  \sum_{l^* \in H^*}\frac{1}{\alpha} \ \sum_{F(l^*)} n_f p_{l^*}\\
    \leq &  \sum_{l^* \in H^*}\frac{1}{\alpha} \ \sum_{F(l^*)} n_f 2 (p_{l^*}-\mu) < \frac{2}{\alpha} \sum_{l^* \in H^*} E[LL(H_{l*}] \\
    < &\ \underset{l\in H^*}{\text{arg max}}  E[LL(H_l)]
\end{align*}

Note that $p_f - \mu < 0$ if flow $f$ takes a path with all good links and $|H^*| \leq \alpha/2$.

\noindent This shows that the greedy inference algorithm will pick a failed link in the first iteration, say $l_1$. Greedily picking the link that maximizes the log likelihood of the hypothesis $\{l_1\}$ is equivalent to deleting all flows crossing $l_1$ and the link $l_1$ itself from the input for analysis for subsequent iterations. Hence, the same proof about iteratively picking a failed link works for subsequent iterations if we ensure that the traffic-skew is maintained after deleting $l_1$ and all flows crossing it. 

For any pair of links $l_2, l_3 \neq l_1$, we have $T(\{l_2, l_3\})/T(\{l_2\}) \leq \frac{1}{\alpha}$. Let's say that the number of packets crossing link $l_2$ is $T'(\{l_2\})$ after we delete $l_1$ and all flows crossing $l_1$. Then we have, 
\begin{align*}
\frac{T'(\{l_2, l_3\})}{T'(\{l_2\})} \leq &\frac{T(\{l_2, l_3\})}{T'(\{l_2\})} = \frac{T(\{l_2, l_3\})}{T(\{l_2\})-T(\{l_1,l_2\})} \\
\leq & \frac{T(\{l_2, l_3\})}{(1-1/\alpha)T(\{l_2\})} = \frac{1}{\alpha-1}
\end{align*}
Thus, for subsequent iterations, after deleting $l_1$, traffic is $1/(\alpha-1)$-skewed and the number of failures is $\alpha/2 - 1$ in the deleted graph. The same argument as before shows that the greedy algorithm will pick a failed link in every iteration as long as there is at least one failed link not in the current hypothesis. \\

\noindent \textbf{Stopping Criteria}: Finally we need to show that once all failed links are picked, the greedy algorithm will halt. This happens when $LL(H_l) < 0$ for all $l$ not in the current hypothesis. Note that the input to log likelihood computations in the current iteration is the topology obtained after deleting all links in the current hypothesis and all flows crossing any of those links. As before we have, 

\begin{align*}
E[LL(H_l)] = & \sum_{f\in F(l)}E[X(f,H)] = \sum_{f\in F(l)}n(p_f - \mu) 
\end{align*}

Note that for a path with all good links, $p_f - \mu \leq p_g - \mu < 0$. Hence, $E[LL(H_l)]$ is bounded away from 0 towards $-\infty$. Since, $LL(H_l)$ is the sum of independent binary variables corresponding to packets each of whose expectation is $\leq (p_g - \mu) < 0$, applying Chernoff bounds followed by a union bound for all links shows that $LL(H_l) < 0$ for all links $l$ with high probability. This complete the proof of Theorem~\ref{greedyGuaranteeConditions}. \end{proof}

\subsection{Defining precision/recall}\label{prDefinition}
\textbf{Precision} is the fraction of predicted failed links that had actually failed and \textbf{recall} is the fraction of failed links that were correctly reported as failed. A faulty device or any of its links are considered to be correct for calculating precision. For calculating recall, including the faulty device itself in $H$ counts as 100\% recall, and including x\% of the device links in $H$ counts as x\% recall, where $H$ is the set of failed links/devices predicted by the algorithm.
More precisely, if $H$ is the set of failed links predicted by the algorithm and $H^*$ is the actual set of failed links, then precision =$|H\cap H^*|/|H|$ and recall = $|H\cap H^*|/|H^*|$. 

We define precision to be 1 if the algorithm returns the empty hypothesis. For 0 actual failures, precision represents the fraction of examples where the algorithm returns a wrong answer and recall is 1 since there are no failures to detect. 

\setlength{\textfloatsep}{2.0pt}
\begin{algorithm}
\begin{algorithmic}[1]
\footnotesize
\Procedure{GreedySearch()}{}
\State current\_hypothesis $\leftarrow$ []
\State $\Delta$ = ComputeInitialDelta()
\While{max($\Delta$) > 0}
    \State link = argmax($\Delta$)
    \State $\Delta$ = UpdateDeltaArr($\Delta$, current\_hypothesis, link)
    \State current\_hypothesis.add(link)
\EndWhile
\State return current\_hypothesis
\EndProcedure
\Procedure{UpdateDeltaArr}{$\Delta$, hypothesis, link}
\State new\_hypothesis $\leftarrow$ hypothesis + [link]
\For{F in FlowsIntersectingWithLink(link)}
    \For{l in F.links}
        \State $\Delta$[l] += GetFlowDelta(new\_hypothesis, l, F)
        \State $\Delta$[l] -= GetFlowDelta(hypothesis, l, F)
    \EndFor
\EndFor
\State return $\Delta$
\EndProcedure
\end{algorithmic}
\caption{Flock inference: Greedy search with JLE (crux)}\label{JLEAlgorithmShort}
\end{algorithm} 

\newpage
\begin{algorithm}
\begin{algorithmic}[1]
\footnotesize
\Procedure{GreedySearch()}{}
\State current\_hypothesis $\leftarrow$ []
\State $\Delta$ = ComputeInitialDelta()
\While{max($\Delta$) > 0}
    \State link = argmax($\Delta$)
    \State $\Delta$ = UpdateDeltaArr($\Delta$, current\_hypothesis, link)
    \State current\_hypothesis.add(link)
\EndWhile
\State return current\_hypothesis
\EndProcedure
\\
\Procedure{UpdateDeltaArr}{$\Delta$, hypothesis, link}
\State new\_hypothesis $\leftarrow$ hypothesis + [link]
\For{F in FlowsIntersectingWithLink(link)}
    \LineComment{\textit{these counters are a simple data structure}} 
    \LineComment{\textit{trick to speed-up the subsequent for loop}}
    \State old\_counters = GetCounters(hypothesis, flow)
    \State new\_counters = GetCounters(new\_hypothesis, flow)
    \For{l in F.links}
        \State $\Delta$[l] += GetFlowDelta(l, *new\_counters)
        \State $\Delta$[l] -= GetFlowDelta(l, *old\_counters)
    \EndFor
\EndFor
\State return $\Delta$
\EndProcedure
\\
\Procedure{GetCounters}{hypothesis, flow}
\State paths\_failed $\leftarrow$ 0
\State num\_paths = dict()
\For{path in flow.paths}
    \If{(PathFailedAsPerHypothesis(path, hypothesis))}
        \State paths\_failed++
    \Else
        \For{link in path}
    	 \State  num\_paths[link]++
    	\EndFor
    \EndIf
\EndFor
\State return (paths\_failed, num\_paths[link], flow.paths.size(), flow.packets\_sent, flow.bad\_packets)
\EndProcedure
\\
\Procedure{GetFlowDelta}{link, paths\_failed, num\_paths, n\_flow\_paths, packets\_sent, bad\_packets}
\State bad\_paths = paths\_failed + num\_paths[link]
\State return GetLogLikelihood(bad\_paths, n\_flow\_paths, packets\_sent, packets\_dropped)
\EndProcedure
\\
\Procedure{ComputeInitialDelta()}{}
\For{l in links}
\State $\Delta$[l] $\leftarrow$ 0
\EndFor
\For{flow in flows}
    \State counters = GetCounters([], flow)
    \For{l in flow.links}
        \State $\Delta$[l] += GetFlowDelta(l, *counters)
    \EndFor
\EndFor
\State return $\Delta$
\EndProcedure
\\
\Procedure{GetLogLikelihood}{bad\_paths, n\_flow\_paths, bad\_packets, packets\_sent}
\State good\_packets = packets\_sent - bad\_packets
\State log\_likelihood = bad\_paths * pow($p_b$, bad\_packets) * \\ \quad \quad \quad \quad \quad \quad \quad \quad \quad \quad \quad pow(1 - $p_b$, good\_packets)
\State good\_paths = n\_flow\_paths - bad\_paths
\State log\_likelihood += good\_paths * pow($p_g$, bad\_packets) \\ \quad \quad \quad \quad \quad \quad \quad \quad \quad \quad \quad * pow(1 - $p_g$, good\_packets)
\State log\_likelihood  /= n\_flow\_paths
\State return log\_likelihood
\EndProcedure
\end{algorithmic}
\caption{Flock's Hypotheses search: Greedy with Joint Likelihood Exploration}
\label{HypothesisSearchAlgorithm}
\end{algorithm}

\begin{algorithm}
\begin{algorithmic}[1]
\footnotesize
\State best\_hypothesis = None
\State max\_likelihood = -$\infty$
\Procedure{SherlockWithJLE()}{}
    \State $\Delta$ = ComputeInitialDelta()
    \State ExploreBranch([], $\Delta$, 0.0)
    \State return best\_hypothesis
\EndProcedure
\Procedure{ExploreBranch}{current\_hypothesis, $\Delta$, current\_likelihood}
\If{current\_likelihood > max\_likelihood}
\State max\_likelihood = current\_likelihood
\State best\_hypothesis = current\_hypothesis
\EndIf
\If{current\_hypothesis.size < K}
    \For{l in links}
        \State new\_hypothesis = current\_hypothesis.add(l)
        \State new\_likelihood = current\_likelihood + $\Delta$[l]
        \State $\Delta_{new}$ = UpdateDeltaArr($\Delta$, current\_hypothesis, l)
        \State ExploreBranch(new\_hypothesis, $\Delta_{new}$, new\_likelihood)
    \EndFor
\EndIf
\EndProcedure
\end{algorithmic}
\caption{JLE can be used to speedup Sherlock's inference, with max concurrent failures = K}\label{JLEBruteForceAlgorithm}
\end{algorithm}

\section{Joint Likelihood Exploration: pseudocode}
Refer to Algorithm~\ref{JLEAlgorithmShort} for a short summary of Flock's inference algorithm and Algorithm~\ref{HypothesisSearchAlgorithm} for full pseudocode.

\section{Full Runtime analysis} \label{fullRuntimeAnalysis}
Let $n$ be the number of links, $m$ be the number of flows, $T$ be an upper bound on the number of links that any flow intersects with, $D$ be an upper bound on the number of flows that any link intersects with and $K$ be the maximum number of concurrent failures (note that \Sys's algorithm doesn't need to know $K$). 

We describe the components in the running time of Flock's overall inference, including Greedy and JLE:
\begin{itemize}[leftmargin=*]
    \item Before the first greedy iteration, the $\Delta$ array is computed once, in a linear pass over all flows and their path sets, incurring $O(n+mT)$ time. 
    \item After each greedy iteration, updating the $\Delta$ array via JLE requires iterating over all flows that intersect with the newly added link. For each such flow $F$, let $L_F$ be the links that $F$ intersects with.
    Updating all entries $\Delta_{H'}(l, F)$ for all $l \in L_F$ requires a couple of passes over $L_F$. 
    Thus, the execution time for subsequent $(K-1)$ greedy iterations, barring the first, is $O((K-1)DT)$. 
\end{itemize}
\noindent Hence, the running time of Greedy inference with JLE is $O(n+mT+(K\!-\!1)DT)$. If we had used just Greedy without JLE (computing likelihood of each hypothesis individually), the runtime would be $O(n + mT + (K\!-\!1)nDT)$. 


We compare runtime with Sherlock.  To compute the MLE, Sherlock scans all $O(n^K)$ hypotheses with $\leq K$ failures. Sherlock is better than brute force, however: it uses $LL(H)$, for an explored hypothesis $H$, to compute $LL(H\oplus l)$ by updating the flow contributions $LL_F(H)$ for all flows $F$ that intersect with link $l$ (since their likelihoods would have changed after flipping the status of link $l$),  giving  $O(n^KDT)$ runtime. We can apply JLE to accelerate Sherlock's inference by evaluating $n$ neighbor hypotheses at once (\ifEurosys Algorithm 3 in ~\cite{FlockFullDraft} \else Algorithm~\ref{JLEBruteForceAlgorithm} in appendix\fi), improving Sherlock's runtime, by a factor of $n$, to $O(n^{K-1}DT)$. However, this is exponential in $K$ and too slow for our purposes. From the analysis above (and our experiments later), it can be seen that Greedy + JLE is dramatically faster.

\fi

\end{document}